\documentclass[times,sort&compress,3p]{elsarticle}

\usepackage{hyperref}

\usepackage{algorithmic}
\usepackage{natbib}
\usepackage{mathrsfs}
\usepackage{graphicx, setspace, amsmath, amssymb, subfigure, url, multirow, booktabs, verbatim, bm,threeparttable, amsthm, color}
\usepackage[linesnumbered,boxed,ruled,commentsnumbered]{algorithm2e}

\newtheorem{theorem}{Theorem}

\newtheorem{condition}{Condition}
\newtheorem{lemma}{Lemma}

\begin{document}

\begin{frontmatter}

\title{
Locally sparse quantile estimation for a partially functional interaction model}



\author[RUC]{Weijuan Liang}
\ead{weijuanliang@yeah.net}

\author[XMU]{Qingzhao Zhang\corref{mycorrespondingauthor}}
\cortext[mycorrespondingauthor]{Corresponding author}
\ead{qzzhang@xmu.edu.cn}

\author[Yale]{Shuangge Ma\corref{mycorrespondingauthor}}
\ead{shuangge.ma@yale.edu}

\address[RUC]{School of Statistics, Renmin University of China, Beijing, China}
\address[XMU]{Department of Statistics and Data Science, School of Economics, The Wang Yanan Institute for Studies in Economics, and Fujian Key Lab of Statistics, Xiamen University, Xiamen, China}
\address[Yale]{Department of Biostatistics, Yale School of Public Health, New Haven, Connecticut, USA}

\begin{abstract}
Functional data analysis has been extensively conducted. In this study, we consider a partially functional model, under which some covariates are scalars and have linear effects, while some other variables are functional and have unspecified nonlinear effects. Significantly advancing from the existing literature, we consider a model with interactions between the {functional and scalar covariates}. To accommodate long-tailed error distributions which are not uncommon in data analysis, we adopt the quantile technique for estimation. To achieve more interpretable estimation, and to accommodate many practical settings, we assume that the functional covariate effects are locally sparse (that is, there exist subregions on which the effects are exactly zero), which naturally leads to a variable/model selection problem. We propose respecting the ``main effect, interaction’’ hierarchy, which postulates that if a subregion has a nonzero effect in an interaction term, then its effect has to be nonzero in the corresponding main functional effect. For estimation, identification of local sparsity, and respect of the hierarchy, we propose a penalization approach. An effective computational algorithm is developed, and the consistency properties are rigorously established under mild regularity conditions. Simulation shows the practical effectiveness of the proposed approach. The analysis of the Tecator data further demonstrates its practical applicability. Overall, this study can deliver a novel and practically useful model and a statistically and numerically satisfactory estimation approach.

\end{abstract}


\end{frontmatter}


\section{Introduction}
Functional data analysis has become routine in statistics. A popular regression setting has a scalar response and functional covariates. In practice, we may directly observe the functional covariates or their realizations at discrete observational (usually time or space) points. In the latter case, estimation of the functional covariates may be first needed. For this regression setting, there have been extensive methodological, computational, and theoretical developments as well as data analyses \cite{aneiros2021variable, cardot2003spline, fan2013functional, yao2005functional}. In particular, both {mean} and robust estimations have been developed \cite{berrendero2019rkhs, shin2016rkhs, tong2018analysis}.

As a natural extension of the aforementioned model, in a partially functional model, there are two types of covariates. The first type of covariates is functional, as described above. In addition, there are also scalar covariates with linear effects. Such a model shares some similar spirit with the partially linear regression \cite{cui2017estimation} but may be more complicated in multiple aspects. As a ``natural next step’’, we further consider the model with interactions between the {functional and scalar covariates}. Interaction is a “basic” concept in data analysis. However, most of the existing interaction analyses are limited to parametric covariate effects. In the literature, there are a handful of studies that examine interactions in the partially linear models \cite{fan2003kernel, wu2014integrative}, and {statistical and computational analysis of such interactions has been shown to be highly nontrivial}. To the best of our knowledge, there has been no interaction analysis with partially functional models that consist of two distinct types of covariate effects.

For the estimation of functional models, both {mean} and  {quantile regression} methods have been developed, accommodating ``regular" and long-tailed error distributions. In this article, we consider data with long-tailed errors and {quantile} estimation, which can be technically more challenging than {mean} estimation. In the existing (both {quantile and mean regression}) studies, it is commonly assumed that the functional covariate effects are smooth. Without additional assumptions/constraints, the estimates are nonzero everywhere. In the past few years, there has been a strong advocacy on locally sparse estimation. Under such an estimation, there exist continual subregions, on which the estimates are exactly zero. In terms of both concept and statistical techniques, this has a strong tie with the sparse estimation for parametric covariate effects. It has been argued that sparse estimation in general can be more interpretable and more reliable. Sparse estimation is ``naturally equivalent to’’ variable/model selection, for which regularization especially penalization techniques have been extensively developed in the past decades. Examples of penalized sparse estimation for functionals include \cite{james2009functional, lin2017locally, zhou2013functional}. 

If there are no interactions in the model, conceptually, some of the existing penalized sparse methods for functionals can be adapted to the partially functional models, although we note that there has been very limited research in this aspect \citep{kong2016partially, ma2019quantile, yao2017regularized}. When interactions are present, however, these methods may lead to a violation of the ``main effect, interaction’’ variable selection hierarchy. This hierarchy has been strongly stressed in the recent parametric interaction analysis studies. Under this hierarchy, if an interaction effect is identified, then one or both of the corresponding main effects have to be identified, corresponding to the weak and strong hierarchy, respectively. It has been argued that in interaction analysis, this hierarchy is statistically sensible and necessary. For the specific model we are interested in, this hierarchy means that, for any specific subregion, if a functional effect is nonzero in an interaction term, then the corresponding main functional effect must be nonzero in this subregion. This brings additional constraints and complexity to estimation. To the best of our knowledge, there is no existing estimation technique that can respect this hierarchy in estimation for our proposed model.

This study may complement and advance the existing literature in multiple important ways. First, a novel model is developed, which can accommodate not only two distinct types of covariate effects but more importantly their interactions. Such extensions are natural and strongly motivated by practical data analysis. This model includes multiple existing models as special cases. Second, we consider {quantile} estimation, which is also motivated by many practical data settings and can be more challenging than {mean} estimation. It is noted that the proposed model and penalized estimation can also be coupled with {mean squares loss function.} Third, locally sparse estimation is conducted, which can lead to more interpretable and more reliable results than those without sparsity. Fourth, as a major advancement, we develop an estimation approach that respects the ``main effect, interaction’’ hierarchy, making this study more aligned with parametric interaction analysis. Last but not least, this study delivers a useful tool for data considered in Section 4 and those alike. Overall, with the significant statistical developments and strong application potential, this study is warranted beyond the existing literature.

The rest of the article is organized as follows. In Section 2, we first describe the data setting, proposed model, and estimation approach. An effective computational algorithm is developed, and statistical properties are then rigorously established. Practical performance of the proposed approach is examined using simulation (Section 3) and data analysis (Section 4). The article concludes with brief discussions in Section 5. Additional theoretical developments and numerical results are presented in the Appendix {and Supplemental Materials}.

\section{Methods}

\subsection{Data and model settings}
Consider a random sample  of size $n$: $\{X_i(t), \bm{z}_i, y_i\}_{i = 1}^n$, where $X_i(t)$ is a functional covariate, $\bm{z}_i =(z_{i1}, \cdots, z_{iq})^\top$ is a $q$-dimensional vector of scalar covariates, and $y_i$ is a scalar response. The proposed model, estimation approach, and statistical and computational properties can be easily extended to data with multiple functional covariates. Assume that $X_i(t), i= 1,\cdots,n$ are independent realizations of an unknown smooth and square-integrable function $X(t)$ on the domain $[0,T]$. Without loss of generality, assume that the functional covariate, scalar covariates, and scalar response have been centered to mean zero.

Consider the partially functional interaction model:
\begin{equation}
\label{eq:1}
y_i=\int_{0}^{T} X_{i}(t) \beta_{0}^*(t) d t+\sum_{k=1}^{q} z_{i k} \int_{0}^{T} X_{i}(t) \beta_{k}^*(t) d t+\sum_{k=1}^{q} z_{i k} \gamma_{k}^*+ \epsilon_i,
\end{equation}
where $\beta_{k}^*(t)$'s {for $k = 0,1,\cdots,q$} are smooth and square-integrable coefficient functions, ${\gamma}_{k}^*$’s are scalar coefficients of $\bm{z}_i$, and the error terms $\epsilon_i$’s are independent of $(X_i(t), \bm{z}_i)$ and satisfy $\text{Pr}(\epsilon_i \leq 0 | X_i(t), \bm{z}_i) = \tau$ for $\tau \in (0,1)$. Note that this assumption accommodates long-tailed error distributions.

As described above, we consider the setting with local sparsity. Take $\beta_{0}^*(t)$ as an example. We say that $\beta_{0}^*(t)$ is locally sparse if there exists a subregion $\mathcal{I} \subset[0, T]$, and $\beta_{0}^*(t)=0$ for all $t\in \mathcal{I}$. Accordingly, $X(t)$ has no contribution to the response for $t\in \mathcal{I}$. Here, we note that there can be more than one region with zero effects, and the region location information is not known a priori. In addition, the proposed approach is flexible enough to also accommodate the case with functional covariate effects being nonzero everywhere. Local sparsity can assist in distinguishing regions with and without effects. And it is easy to see the natural connection with variable selection for parametric models.

The proposed model is more complicated than some existing alternatives as local sparsity may apply to both the main effect and interactions. With the connection between local sparsity and variable selection, we naturally encounter the ``main effect, interaction’’ variable selection hierarchy. In our analysis, we are not interested in the sparsity in $\gamma^*_k$’s (although we note that extending to accommodate potential sparsity in $\gamma^*_k$’s is relatively easy with the parametric nature). As such, the hierarchy boils down to the relationship between the subregions of $\beta_k^*$’s $(k=1, \ldots, q)$ with zero/nonzero effects and those of $\beta_0^*$. More specifically, we say that the hierarchy is satisfied if and only if for any subregion $\mathcal{I} \subset[0, T]$, if $\beta_{0}^*(t) = 0$ for all $t\in \mathcal{I}$, then $\beta_{k}^*(t) = 0$. We note that a more rigorous definition should rule out any measure zero set.

\subsection{Estimation}

For estimating the unknown parameters, we propose minimizing the penalized objective function:
\begin{equation}
\label{eq:2}
\begin{aligned}
Q(\bm{\beta}(t), \bm{\gamma}) = & \frac{1}{n}\sum_{i=1}^n \rho_\tau \left(y_i -\int_{0}^{T} X_{i}(t) \beta_{0}(t) d t - \sum_{k=1}^{q} z_{i k} \int_{0}^{T} X_{i}(t) \beta_{k}(t) d t - \sum_{k=1}^{q} z_{i k} \gamma_{k} \right)\\ 
&+ \sum_{k = 1}^q \frac{\kappa}{T}  \int_0^T p_{\lambda_1}(|\beta_{k}(t)|) dt +  \frac{\kappa}{T} \int_0^T p_{\lambda_2} (\| \bm{\beta}(t) \|_2) dt + \eta\sum_{k=0}^q \int_0^T \beta^{\prime\prime2}_{k}(t) dt,
\end{aligned}
\end{equation}
where  $\rho_\tau(u) = u(\tau - I(u<0))$ is the {quantile loss} function, $\| \bm{\beta}(t) \|_2 = (\sum_{k=0}^q \beta_k^{2}(t))^{1/2}$,  $p_{\lambda_j}(\cdot)$’s are penalty functions with tuning parameters $\lambda_j$’s (for $j=1, 2$), $\kappa$ is a modifier and will be discussed below, $\eta$ is a tuning parameter, and $\beta_{k}^{\prime\prime}(t)$ is the second-order derivative of $\beta_{k}(t)$ with respect to $t$.
Various penalty functions can be adopted here, and $p_{\lambda_1}(\cdot)$ and $p_{\lambda_2}(\cdot)$ do not need to be the same. In our theoretical and numerical developments, we adopt MCP \cite{zhang2010nearly} for both $p_{\lambda_1}(\cdot)$ and $p_{\lambda_2}(\cdot)$, where $p_{\lambda_j}(t) = \lambda_j \int_0^{|t|} \left(1-\frac{x}{\lambda_j\xi} \right)_+dt$, $\lambda_j \geq 0$, and $\xi > 0$ is a regularization parameter. It is expected that, with SCAD and some other penalties, properties will be similar.

In (\ref{eq:2}), the first term is a ``standard’’ lack-of-fit based on the quantile technique. Under the smoothness assumption, the last penalty on derivative has been routinely adopted. Here we note that a stronger smoothness assumption and correspondingly a higher order derivative can also be adopted. The most significant and innovative advancement is the first and second penalty terms. In ``ordinary’’ locally sparse estimation, penalties similar to $\int_0^T p_{\lambda_1}(|\beta_{k}(t)|) dt$ have been adopted \cite{lin2017locally}. In our estimation, new challenges are brought by the hierarchy. Motivated by the sparse group penalization for parametric models \cite{liu2013sparse}, we treat $(\beta_0, \beta_1, \ldots, \beta_q)$ as a ``group’’. For a subregion, $\frac{\kappa}{T} \int_0^T p_{\lambda_2} (\| \bm{\beta}(t) \|_2) dt$ determines whether this group of functionals has overall zero effect. If not, then $\sum_{k = 1}^q \frac{\kappa}{T}  \int_0^T p_{\lambda_1}(|\beta_{k}(t)|) dt$ determines which of the $q$ interaction effects are nonzero. Note that here no penalty is applied to $\beta_0(t)$, ensuring that the corresponding estimate is nonzero, and hence the hierarchy is guaranteed.

Directly optimizing (\ref{eq:2}) is challenging with the infinite dimension of the unknown functionals. Here we adopt a popular B-spline expansion-based technique, {which can be preferred with its compact support property, computational efficiency, and satisfactory performance with capturing local sparsity}. Denote $\mathcal{H}_{dM_n}$ as the linear space spanned by a set of order $d + 1$ B-spline basis functions $B_1(t), \cdots , B_{M_n+d}(t)$, each with  $M_n + 1$ equally spaced knots $0 = t_0 < t_1 < \cdots < t_{M_n} = T$ in the domain $[0,T]$. In (\ref{eq:2}), second-order derivatives are taken, corresponding to $d=2$. We refer to \cite{de1978practical} for the construction of B-spline basis functions and related. 
Denote $\bm{B}(t) = (B_1(t), \cdots, B_{M_n+d} (t))^\top$. 
Then we parameterize coefficient functions $\beta_k(t)=\bm{B}(t)^\top \bm{b}_k$ for $k=0, \cdots, q$, where $\bm{b}_k = (b_{k,1}, \cdots, b_{k,M_n+d})^\top$. Let $\bm{Z}=(\bm{z}_1, \cdots, \bm{z}_n)^\top$, $\bm{X}=\left(\bm{x}_{1}, \ldots, \bm{x}_{n}\right)^{\top}$ be the $n \times(M_{n}+d)$ matrix with the $(i,j)$th entry being $x_{i j}=\int_{0}^{T} X_{i}(t) B_{j}(t) d t$, and $\bm{U} =\left(\bm{u}_{1}, \ldots, \bm{u}_{n}\right)^{\top}$ be the $n\times q(M_n+d)$ matrix with $\bm{u}_{i}= \bm{z}_{i} \otimes \bm{x}_{i}$, where $\otimes$ is the Kronecker product. Further denote $\bm{\Psi} = (\bm{X}, \bm{U})$, which is 
$n \times q_n$ with $q_n = (q+1)\times(M_n+d)$.

The first term of (\ref{eq:2}) can be rewritten as: 
\begin{equation}
\label{eq:3}
\frac{1}{n} \sum_{i = 1}^n \rho_\tau (y_i - \bm{\psi}_i^\top \bm{b} -\bm{z}_i^\top \bm{\gamma}),
\end{equation}
where $\bm{b} = (\bm{b}_{0}^\top, \cdots, \bm{b}_{q}^\top)^\top$ and $\bm{\gamma} = (\gamma_1, \cdots, \gamma_q)^\top$.
In Lemma \ref{lemma:1}  (Appendix), we examine approximating the sparse group penalty under this basis expansion. In particular, setting the modifier $\kappa = M_n$, we  have:
\begin{equation}
\begin{aligned}
\label{eq:4}
&\sum_{k = 1}^q \frac{M_n}{T} \int_{0}^{T} p_{\lambda_1}(|\beta_{k}(t)|) dt + \frac{M_n}{T} \int_{0}^{T} p_{\lambda_2}(\|\bm{\beta}(t)\|_2) dt \\
&~~~~~ \approx \sum_{k = 1}^q \sum_{l=1}^{M_{n}} p_{\lambda_1}\left( \| \bm{b}_{k}\|_{\bm{W}_l}\right) + \sum_{l=1}^{M_{n}} p_{\lambda_2}\left( \| \bm{b}\|_{\bm{W}_l}\right) ,
 \end{aligned}
\end{equation}
where $\bm{W}_{l}$ is the $\left(M_{n}+d\right) \times\left(M_{n}+d\right)$ matrix with the $(i,j)$th entry $w_{lij}= \frac{M_n}{T}\int_{t_{l-1}}^{t_{l}} B_{i}(t) B_{j}(t) d t$ if $l \leq$ $i, j \leq l+d$, and $w_{lij}=0$ otherwise, 
$\| \bm{b}_{k}\|_{\bm{W}_l} = ( \bm{b}_{k}^{\top} \bm{W}_{l} \bm{b}_{k})^{1/2}$, and $\|\bm{b} \|_{\bm{W}_l} = ({ \sum_{k=0}^q \bm{b}_{k}^{\top} \bm{W}_{l} \bm{b}_{k}})^{1/2}$. Let $\bm{V}$ be the $(M_n+d) \times (M_n + d)$ matrix with the $(i,j)$th entry $v_{i j}=\int_{0}^{T} \frac{d^{2} B_{i}(t)}{d t^{2}} \frac{d^{2} B_{j}(t)}{d t^{2}} d t$.  Then,
\begin{equation}
\label{eq:5}
 \sum_{k = 0}^q \int_{0}^{T} \beta_{k}^{\prime\prime2}(t) dt = \sum_{k = 0}^q  \bm{b}_{k}^{\top} \bm{V} \bm{b}_{k}.
\end{equation}

With (\ref{eq:3}), (\ref{eq:4}) and (\ref{eq:5}), we propose estimating $(\bm{b}^*, \bm{\gamma}^*)$ by minimizing the following objective function: 
\begin{equation}
\label{eq:6}
\begin{aligned}
Q(\bm{b}, \bm{\gamma})=& \frac{1}{n} \sum_{i=1}^n \rho_\tau \left(y_i - \bm{\psi}_i^\top \bm{b} -\bm{z}_i^\top \bm{\gamma} \right)+\sum_{k=1}^{q} \sum_{l=1}^{M_{n}} p_{\lambda_{1}}\left(  \|\bm{b}_k \|_{\bm{W}_l}\right) +\sum_{{l}=1}^{M_{n}} p_{\lambda_{2}}\left( \|\bm{b} \|_{\bm{W}_l} \right)   + \eta \sum_{k=0}^{q} \bm{b}_{k}^{\top} \bm{V} \bm{b}_{k}.
\end{aligned}
\end{equation}
Denote $(\hat{\bm{b}}, \hat{\bm{\gamma}})$ as the minimizer. Then the estimate of $\beta_k^*(t)$ is $\hat{\beta}_k(t) = \bm{B}^\top(t) \hat{\bm{b}}_k$.

\subsection{Computation}

To accommodate the non-differentiable quantile loss function, we resort to the majorize-minimization (MM) technique. In addition, we adopt the local quadratic approximation (LQA) technique for the sparse group penalty. 

The proposed algorithm is iterative. At the $(m+1)$th iteration, with estimate $\bm{b}_k^{(m)}$ from the $m$th iteration, we have:
\begin{equation*}
\begin{aligned}
p_{\lambda_1} \left( \|\bm{b}_k\|_{\bm{W}_l} \right)
&\approx p_{\lambda_1} ( \|\bm{b}_k^{(m)}\|_{\bm{W}_l}) +  \frac{1}{2} \frac{p_{\lambda_1}^{\prime}( \|\bm{b}_k^{(m)}\|_{\bm{W}_l})}{\|\bm{b}_k^{(m)}\|_{\bm{W}_l}}(\|\bm{b}_k\|_{\bm{W}_l}^2 - \|\bm{b}_k^{(m)}\|_{\bm{W}_l}^2)\\
& = \frac{1}{2} \frac{p_{\lambda_1}^{\prime}( \|\bm{b}_k^{(m)}\|_{\bm{W}_l})}{\|\bm{b}_k^{(m)}\|_{\bm{W}_l}}\|\bm{b}_k\|_{\bm{W}_l}^2 + G_0(\bm{b}_k^{(m)}),
\end{aligned}
\end{equation*}
where $p^\prime_{\lambda_1}(t) = \lambda_1 (1-|t|/(\lambda_1\xi))_+$ is the first-order derivative of $p_{\lambda_1}(t)$, and $G_0(\bm{b}_k^{(m)})$ is a function of $\bm{b}_k^{(m)}$ and does not depend on $\bm{b}_k$. As such, we can obtain the LQA approximation of the sparse group penalty as:
\begin{equation}
\label{eq:7}
\begin{aligned}
 &\sum_{k=1}^q \sum_{l=1}^{M_n} p_{\lambda_1} \left( \|\bm{b}_k\|_{\bm{W}_l} \right) + \sum_{l=1}^{M_n} p_{\lambda_2} \left( \|\bm{b}\|_{\bm{W}_l}\right)  \approx \bm{b}^\top \breve{\bm{W}}^{(m)}\bm{b} +  G_1(\bm{b}^{(m)}),
 \end{aligned}
\end{equation}
where $\breve{\bm{W}}^{(m)} = \text{diag}(\breve{\bm{W}}^{(m)}_0, \cdots, \breve{\bm{W}}^{(m)}_q)$ is the $q_n \times q_n$ block diagonal matrix with:
$$
\breve{\bm{W}}^{(m)}_0 = \frac{1}{2} \sum_{l=1}^{M_n} \frac{p_{\lambda_2}^{\prime}\left(\|\bm{b}^{(m)}\|_{\bm{W}_l} \right)}{\|\bm{b}^{(m)}\|_{\bm{W}_l}} \bm{W}_{l},
$$ 
and 
$$
\breve{\bm{W}}^{(m)}_k = \frac{1}{2} \sum_{l=1}^{M_n} \left(\frac{p_{\lambda_1}^{\prime}(\|\bm{b}_k^{(m)}\|_{\bm{W}_l} )}{\|\bm{b}_k^{(m)}\|_{\bm{W}_l}}+  \frac{p_{\lambda_2}^{\prime}\left(\|\bm{b}^{(m)}\|_{\bm{W}_l} \right)}{\|\bm{b}^{(m)}\|_{\bm{W}_l}} \right) \bm{W}_{l}, 
$$
for $k = 1,\cdots, q$, and $G_1(\bm{b}^{(m)})$ is free of $\bm{b}$.

Let $\bm{\Phi}$ be the $n\times d_n$ matrix with the $i$th row being $\bm{\phi}_i = (\bm{x}_i^\top, \bm{u}_i^\top, \bm{z}_i^\top)^\top \in \mathcal{R}^{d_n}$ and $d_n = q_n + q$. Denote the coefficient vector as $\bm{\omega} = (\bm{b}^\top, \bm{\gamma}^\top)^\top \in \mathcal{R}^{d_n}$. Let $\tilde{\bm{V}} = \text{diag} (\bm{V},\cdots, \bm{V}, \bm{0}_q)$ and  $\tilde{\bm{W}}^{(m)}=\text{diag}(\breve{\bm{W}}^{(m)}, \bm{0}_{q})$ be $d_n \times d_n$ block-diagonal matrices, where $\bm{0}_{q}$ is the $q \times q$ matrix with all entries being 0.

With the MM algorithm, at the $(m+1)$th iteration, given the residual value $\bm{r}^{(m)} = \bm{y}-\bm{\Phi} \bm{\omega}^{(m)}$, the quantile loss is majorized at $\bm{r}^{(m)}=(r_1^{(m)}, \cdots, r_n^{(m)})^\top$ by the quadratic function:
$$
\xi (\bm{r} | \bm{r}^{(m)} )=\frac{1}{n} \sum_{i=1}^{n} \frac{1}{4}\left(\frac{r_{i}^{2}}{\varrho+|r_{i}^{{(m)}}|}+(4 \tau-2) r_i+c\right),
$$
where $\bm{r} = (r_1, \cdots, r_n)^\top$, $\varrho$ is a small perturbation, and $c$ is a constant. 

The overall objective function at the $(m+1)$th iteration is: 
$
\tilde{Q} (\bm{\omega} | \bm{\omega}^{(m)})= \xi (\bm{r} | \bm{r}^{(m)}) +\bm{\omega}^{\top} \tilde{\bm{W}}^{(m)} \tilde{\bm{ \omega}}+\eta \bm{\omega}^{\top} \tilde{\bm{V}} \bm{\omega}.
$
The first-order derivative of $\tilde{Q} (\bm{\omega} | \bm{\omega}^{(m)})$ with respect to $\bm{\omega}$ is:
\begin{equation*}
\begin{aligned}
\tilde{Q}^\prime (\bm{\omega} | \bm{\omega}^{(m)}) & =\frac{1}{2 n} \sum_{i=1}^{n} \bm{\phi}_{i}\left(1-2 \tau-{r_{i}}/{(\varrho+|r_{i}^{(m)}|)}\right)+2 \tilde{\bm{W}}^{(m)} \bm{\omega}+2 \eta \tilde{\bm{V}} \bm{\omega} \\
& =\frac{1}{2 n} \bm{\Phi}^{\top} v_{\varrho}(\bm{\omega} | \bm{\omega}^{(m)})+2 \tilde{\bm{W}}^{(m)} \bm{\omega}+2 \eta \tilde{\bm{V}} \bm{\omega},
\end{aligned}
\end{equation*}
where 
$$ v_{\varrho}(\bm{\omega}| \bm{\omega}^{(m)}) = \left(1-2 \tau-\frac{r_{1}}{\varrho+|r_{1}^{(m)}|}, \cdots, 1-2 \tau-\frac{r_{n}}{\varrho+|r_{n}^{(m)}|} \right)^\top$$ 
is a length-$n$ column vector.
The second-order derivative of $\tilde{Q} (\bm{\omega} | \bm{\omega}^{(m)})$ with respect to $\bm{\omega}$ is:
\begin{equation*}
\begin{aligned}
\tilde{Q}^{\prime \prime} (\bm{\omega} | \bm{\omega}^{(m)}) =\frac{1}{2 n} \sum_{i=1}^{n} \frac{\bm{\phi}_{i} \bm{\phi}_{i}^{\top}}{\varrho+| r^{(m)}_{i}|}+2 \tilde{\bm{W}}^{(m)}+2 \eta \tilde{\bm{V}} =\frac{1}{2 n} \bm{\Phi}^{\top} \bm{R}^{(m)} \bm{\Phi}+2 \tilde{\bm{W}}^{(m)}+2 \eta \tilde{\bm{V}},
\end{aligned}
\end{equation*}
where 
$$\bm{R}^{(m)}=\text{diag}\left(\frac{1}{\varrho+| r_{1}^{(m)}|}, \cdots, \frac{1}{\varrho+| r_{n}^{(m)}|} \right)$$
is an $n \times n$ diagonal matrix. Then the Gauss-Newton step direction is: 
\begin{eqnarray}
\label{eq:9}
&&\Delta_{\varrho}^{(m)}( \bm{\omega}|\bm{\omega}^{(m)})  =-\left[ \tilde{Q}^{\prime \prime}(\bm{\omega} | \bm{\omega}^{(m)})\right]^{-1}{\tilde{Q}^{\prime} (\bm{\omega} | \bm{\omega}^{(m)})}\nonumber\\
&=&-\left(\bm{\Phi}^{\top} \bm{R}^{(m)} \bm{\Phi}+4 n \tilde{\bm{W}}^{(m)}+4 n \eta \tilde{\bm{V}}\right)^{-1}\left(\bm{\Phi}^{\top} v_{\varrho}(\bm{\omega} | \bm{\omega}^{(m)} )+4 n \tilde{\bm{W}}^{(m)} \bm{\omega}+4 n \eta \tilde{\bm{V}} \bm{\omega}\right).
\end{eqnarray}

Overall, the proposed computational algorithm proceeds as follows:
\begin{enumerate}[Step 1.]
 \item Initialize $\hat{\bm{\omega}}$ as $\hat{\bm{\omega}}^{(0)} = (\bm{\Phi}^\top\bm{\Phi} + n\eta \tilde{\bm{V}})^{-1} \bm{\Phi}^\top \bm{y}$ and $m = 0$.
\item Given $\hat{\bm{\omega}}^{(m)}$, compute  $\tilde{\bm{W}}^{(m)}$, $\bm{R}^{(m)}$, and $v_\varrho(\hat{\bm{\omega}}^{(m)})$. Update:
$$\hat{\bm{\omega}}^{(m+1)} = \hat{\bm{\omega}}^{(m)} + \Delta_\varrho^{(m)} (\hat{\bm{\omega}}^{(m)}), \text{ and } m = m+1.$$
\item Repeat 2 until convergence, which is concluded if the norm of the difference between the estimates from two consecutive iterations is smaller than a prespecified cutoff.  The final estimate of $\bm{\omega}$ is obtained by further setting the elements of $\hat{\bm{\omega}}$ with absolute values smaller than a prespecified threshold to zero.
\end{enumerate}

This algorithm is built on the MM and LQA techniques, both of which have been well examined in published literature. Convergence of the algorithm can be established following the literature and is achieved in all of our numerical studies. With the LQA, finite iterations cannot lead to sparse estimation. Following published studies, a cutoff (whose value is not crucial) is imposed in Step 3. In our numerical study, we use $10^{-3}$. As in the literature, the value of $M_n$ is also not crucial since the smoothness of estimation is controlled by the roughness penalty, as opposed to the number of knots.  In our numerical study, {
we use cubic B-splines with 71 equally spaced knots to estimate $\bm{\beta}(t)$'s, following \cite{lin2017locally}. Note that other knot placement strategies can be considered such as some data-driven methods putting knots at certain quantiles of covariates.} Following \citep{shi2014penalized, wu2020structured}, we set $\lambda_2 = \sqrt{q+1} \lambda_1$ and $\xi = 6$ and perform a grid search for the optimal $(\eta, \lambda_1)$ based on prediction performance. More details are provided in the numerical studies below. {The R code implementing the proposed algorithm is publicly available at  https://github.com/weijuanliang12138/SHLoS-R-Code.}

\subsection{Theoretical properties}
Let $f_i(\cdot)$ and $F_i(\cdot)$ be the probability density function and distribution function of $\epsilon_i$ given $(X_i(t), \bm{z}_i)$, respectively. Denote $\bm{B}_n= \mbox{diag}\{f_1(0), \cdots, f_n(0)\}$. We assume the following conditions.

\begin{condition}
For $i = 1, \cdots, n$, in a neighborhood of zero, $f_i$ is continuous and satisfies $0 < c \leq  f_i \leq C < \infty$, where $c$ and $C$ are constants. In addition, the first-order derivative $f_i^\prime$ has a uniform upper bound.
\end{condition}

\begin{condition}
For $k = 0, \cdots, q$, $\beta_k^*(t)$ belongs to the $H\ddot{o}lder$ space $C^{\alpha, \nu}([0,1])$. Specifically,  $|\beta^{*(\alpha)}_{k}(x_1) - \beta^{*(\alpha)}_{k}(x_2)| \leq C_1|x_1 - x_2|^\nu$ for a constant $C_1$, positive integer $\alpha$, and $\nu\in(0, 1]$, and for all $0\leq x_1, x_2 \leq 1$, where $\beta^{*(\alpha)}_k(\cdot)$ is the $\alpha$th-order derivative of $\beta^{*}_k(\cdot)$. Let $r = a+\nu$. Assume $r > 1.5$.
\end{condition}

\begin{condition}
There exist positive constants $C_2$ and $C_3$, such that $({ \int_0^T X(t)^2 dt})^{1/2}  \leq C_2 < \infty$ and $ |{z}_{k}|  \leq C_3 < \infty$ for $k = 1, \cdots, q$. In addition, there exist  positive constants $C_4$, $C_5$, $C_6$, and $C_7$, such that
$$  C_4 M_n^{-1} \leq \lambda_{\min}(n^{-1} \bm{\Psi} \bm{\Psi}^\top ) \leq  \lambda_{\max}(n^{-1} \bm{\Psi} \bm{\Psi}^\top)\leq  C_5 M_n^{-1}, $$
 $$  C_6 \leq   \lambda_{\min} (n^{-1} \check{\bm{Z}}\check{\bm{Z}}^\top) \leq \lambda_{\max} (n^{-1} \check{\bm{Z}}\check{\bm{Z}}^\top)  \leq C_7,$$
 where $\check{\bm{Z}} = (\bm{I}_n - \bm{\Psi}(\bm{\Psi}^\top \bm{B}_n \bm{\Psi})^{-1}\bm{\Psi}^\top \bm{B}_n ) \bm{Z} $.  

\end{condition}

\begin{condition}
$M_n= O(n^{\frac{1}{2r+1}})$.
\end{condition}

Condition 1 is common in the quantile regression literature and weaker than those assumed with {mean} estimations. Condition 2 ensures that there exists $\bm{b}_k^*\in\mathcal{R}^{M_n+d}$ such that $\sup_{t\in[0,T]} |\beta_k^*(t)-\bm{B}(t)^\top \bm{b}_k^*| = O(M_n^{-r})$ for $k=0,\cdots, q$ \cite{schumaker2007spline}.  Condition 3 is on the covariates and design matrices, which is analogous to those in \cite{zhou2013functional} and \cite{sherwood2016partially}.  Condition 4 is also common in the spline literature.  

Denote the null region of $\beta_k^*(t)$ as $\mathcal{N}_k=\{t \in [0,T]:\beta_k^*(t)=0 \}$. The asymptotic properties of the proposed estimator can be summarized as follows. 

\begin{theorem}
Under Conditions 1-4, if  $n^{-\frac{r}{2r+1}}/\min(\lambda_1, \lambda_2) = o(1)$, $\max(\lambda_1, \lambda_2) = o(1)$ and $\eta = o(n^{-1/2})$, then there exists a local minimizer $(\hat{\bm{b}}, \hat{\bm{\gamma}})$ of (\ref{eq:6}), such that for all $k = 0, \cdots, q$ with $\hat{\beta}_k(t) = \bm{B}^\top(t) \hat{\bm{b}}_k$,  \\
(1) $\int_0^T(\hat{\beta}_k(t) - \beta_k^*(t))^2 dt = O_p(n^{-2r/(2r+1)})$   and  $\|\hat{\bm{\gamma}} - \bm{\gamma}^* \|_2 = O_p(n^{-1/2})$,\\
(2)  $\hat{\beta}_k(t) = 0$ for all $t\in \mathcal{N}_k$ with probability tending to one.
\end{theorem}

Proof is provided in the Appendix. This theorem establishes the estimation and selection consistency properties. It is observed that the convergence rate of $\hat{\beta}_{k}(t)$ is $n^{-r/(2r+1)}$, which is optimal \cite{stone1985additive}. The convergence rate of $\hat{\bm{\gamma}}$ is free of $M_n$ -- the optimal root-$n$ rate is achieved. The selection consistency holds by result (2). With the design of the penalty, the ``main effect, interaction’’ hierarchy is automatically satisfied.

\section{Simulation}
{Data is generated from the following model:}
\begin{equation}
\label{eq:10}
y_{i}=\int_{0}^{1} X_{i}(t) \beta_0^*(t) d t+\sum_{k=1}^{2} z_{i k} \int_{0}^{1} X_{i}(t) \beta_k^*(t) d t+\sum_{k=1}^{2} z_{i k} \gamma_{k}^*+\epsilon_{ i}.
\end{equation}
{We consider three different scenarios of coefficient functions corresponding to various levels of sparsity and number of null regions. All of these functions in each scenario satisfy the ``main effect, interaction" hierarchy.}

{
Scenario I:  60\% regions of $\beta_{10}^*(t)$ have contribution to the response, and there is a null region in  $\beta_{10}^*(t)$. The functional main effect is:
\begin{equation*}
\beta_{10}^*(t) = 
\left\{
             \begin{array}{ll}
             2(1-t) \sin(2\pi(t+0.2))  &  0\leq t \leq 0.3, \\
             0  & 0.3< t < 0.7, \\
             2t \sin(2\pi(t-0.2)) &  0.7 \leq t \leq1.
             \end{array}
\right.
\end{equation*}
For the functional interaction effects, we consider: (1) $\beta_{11}^*(t) = \beta_{10}^*(t) \text{ for } t \in [0, 0.3]$, and $\beta_{11}^*(t) = 0$ otherwise, (2) $\beta_{12}^* (t) = \beta_{10}^*(t) \text{ for } t \in [0.7, 1]$, and $\beta_{12}^*(t) = 0$ otherwise. These functions are demonstrated in Figure \ref{Fig:1} by black solid lines. }

{
Scenario II: 30\% regions of the main effect are nonnull regions, and there are four null regions on the entire domain of $\beta_{20}^*(t)$. The functional main effect and interactions are defined as:
\begin{equation*}
\beta_{20}^*(t) = 
\left\{
             \begin{array}{ll}
             5\sin(10\pi(t-0.2)) &  0.2< t \leq 0.3, \\
             -3\sin(10\pi(t-0.5)) & 0.5< t \leq 0.6,\\
             3.5\sin(10\pi(t-0.7)) &  0.7< t \leq 0.8,\\
             0  & \text{otherwise},
             \end{array}
\right.
\end{equation*}
\begin{equation*}
\beta_{21}^*(t) = 
\left\{
             \begin{array}{ll}
             2(t-0.25)^2/0.05^2 -2 &  0.2 <t \leq 0.3, \\
             5\sin(10\pi(t-0.5)) & 0.5 < t \leq0.6,\\
              0  & \text{otherwise},
             \end{array}
\right.
\end{equation*}
and
\begin{equation*}
\beta_{22}^*(t) = 
\left\{
             \begin{array}{ll}
             2.5\sin(10\pi(t-0.5)) & 0.5< t \leq 0.6,\\
             4(t-0.75)^2/0.05^2 -4 &  0.7< t \leq 0.8, \\
              0  & \text{otherwise},
             \end{array}
\right.
\end{equation*}
respectively. These functions are presented in Figure \ref{Fig:2} by black solid lines.}

{
Scenario III: 17.5\% regions of $\beta_{30}^*(t)$ have nonzero effects on the response, and there are eight null regions on the entire domain of the main effect. The functional main effect and interactions are defined as:
\begin{equation*}
\beta_{30}^*(t) = 
\left\{
             \begin{array}{ll}
             4(t-0.1375)^2/0.0125^2 -4    & 0.125 < t \leq 0.15,\\
             7\sin(40\pi(t-0.175))      & 0.175 < t \leq 0.2,\\
             -6\sin(40\pi(t-0.325))     & 0.325 < t \leq 0.35,\\
             8\sin(40\pi(t-0.6))          & 0.6     < t \leq 0.625,\\
             -10\sin(40\pi(t-0.7))       & 0.7     < t \leq 0.725,\\
             5\sin(40\pi(t-0.8))          & 0.8     < t \leq 0.825,\\
             -7\sin(40\pi(t-0.875))     & 0.875 < t \leq 0.9,\\
             0  & \text{otherwise},
             \end{array}
\right.
\end{equation*}
\begin{equation*}
\beta_{31}^*(t) = 
\left\{
             \begin{array}{ll}
             10\sin(40\pi(t-0.125))    & 0.125 < t \leq 0.15,\\
             6\sin(40\pi(t-0.325))      & 0.325 < t \leq 0.35,\\ 
             8(t-0.7125)^2/0.0125^2 -8    & 0.7     < t \leq 0.725,\\
             9\sin(40\pi(t-0.875))      & 0.875 < t \leq 0.9,\\
              0  & \text{otherwise},
             \end{array}
\right.
\end{equation*}
and
\begin{equation*}
\beta_{32}^*(t) = 
\left\{
             \begin{array}{ll}
             5\sin(40\pi(t-0.175))                        & 0.175 < t \leq 0.2,\\
             10(t-0.6125)^2/0.0125^2 -10          & 0.6     < t \leq 0.625,\\
             7\sin(40\pi(t-0.8))                            & 0.8     < t \leq 0.825,\\
              0  & \text{otherwise},
             \end{array}
\right.
\end{equation*}
respectively. These functions are demonstrated in Figure \ref{Fig:3} by black solid lines.
}

\begin{figure}[!h]
\centering
\includegraphics[ height=3.8in, width=5.2in]{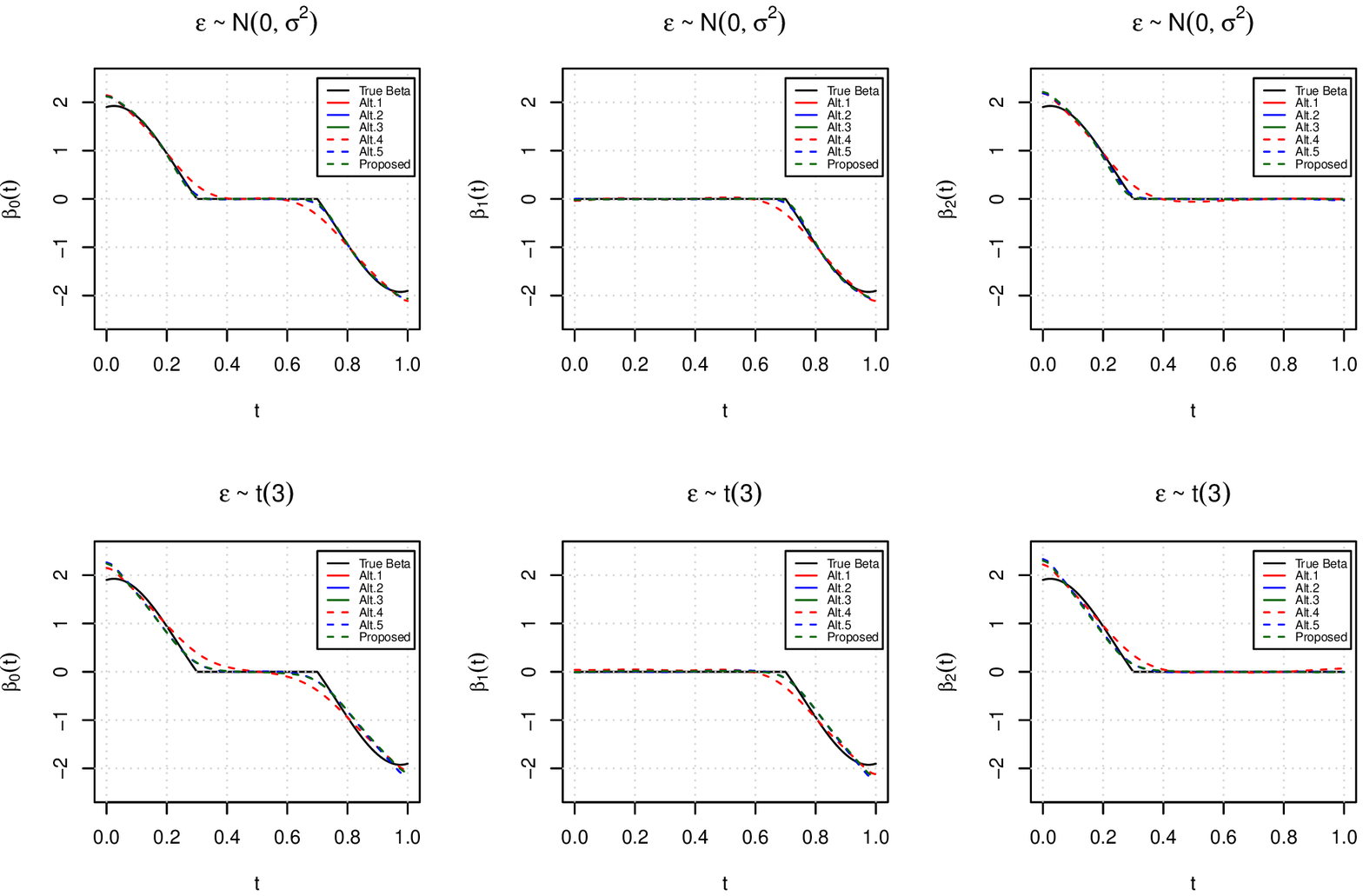}
\caption{{Average of $\hat{\beta}(t)$'s in Scenario I  with $n = 300$ based on 100 replicates under Case 1 (top) and Case 2 (bottom), respectively. Left/middle/right: $\beta_{0}(t)$/$\beta_{1}(t)$/$\beta_2(t)$.}}
\label{Fig:1}
\end{figure}

\begin{figure}[!h]
\centering
\includegraphics[ height=3.8in, width=5.2in]{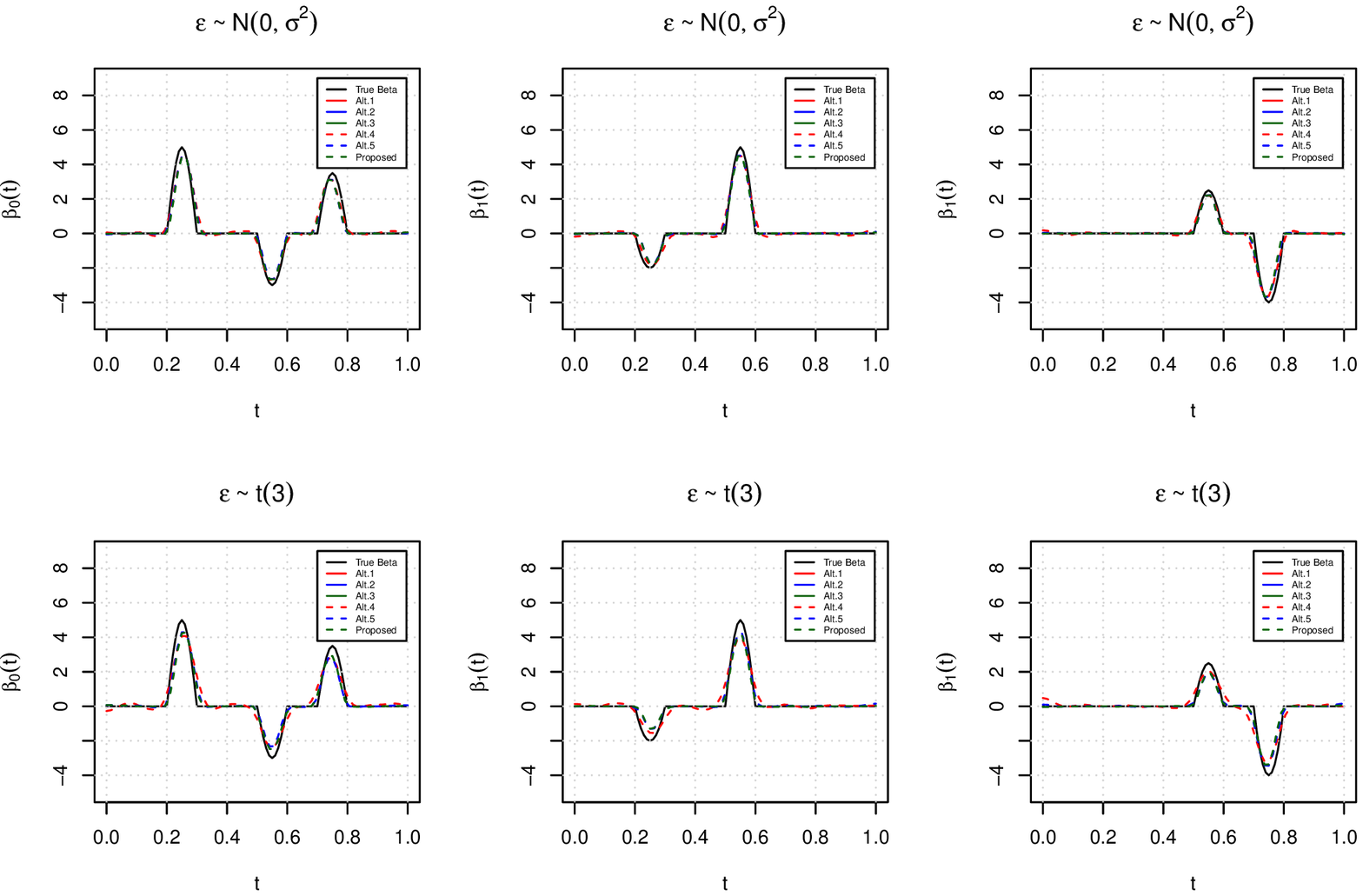}
\caption{{Average of $\hat{\beta}(t)$'s in Scenario II  with $n = 300$ based on 100 replicates under Case 1 (top) and Case 2 (bottom), respectively. Left/middle/right: $\beta_{0}(t)$/$\beta_{1}(t)$/$\beta_2(t)$.}}
\label{Fig:2}
\end{figure}

\begin{figure}[!h]
\centering
\includegraphics[ height=3.8in, width=5.2in]{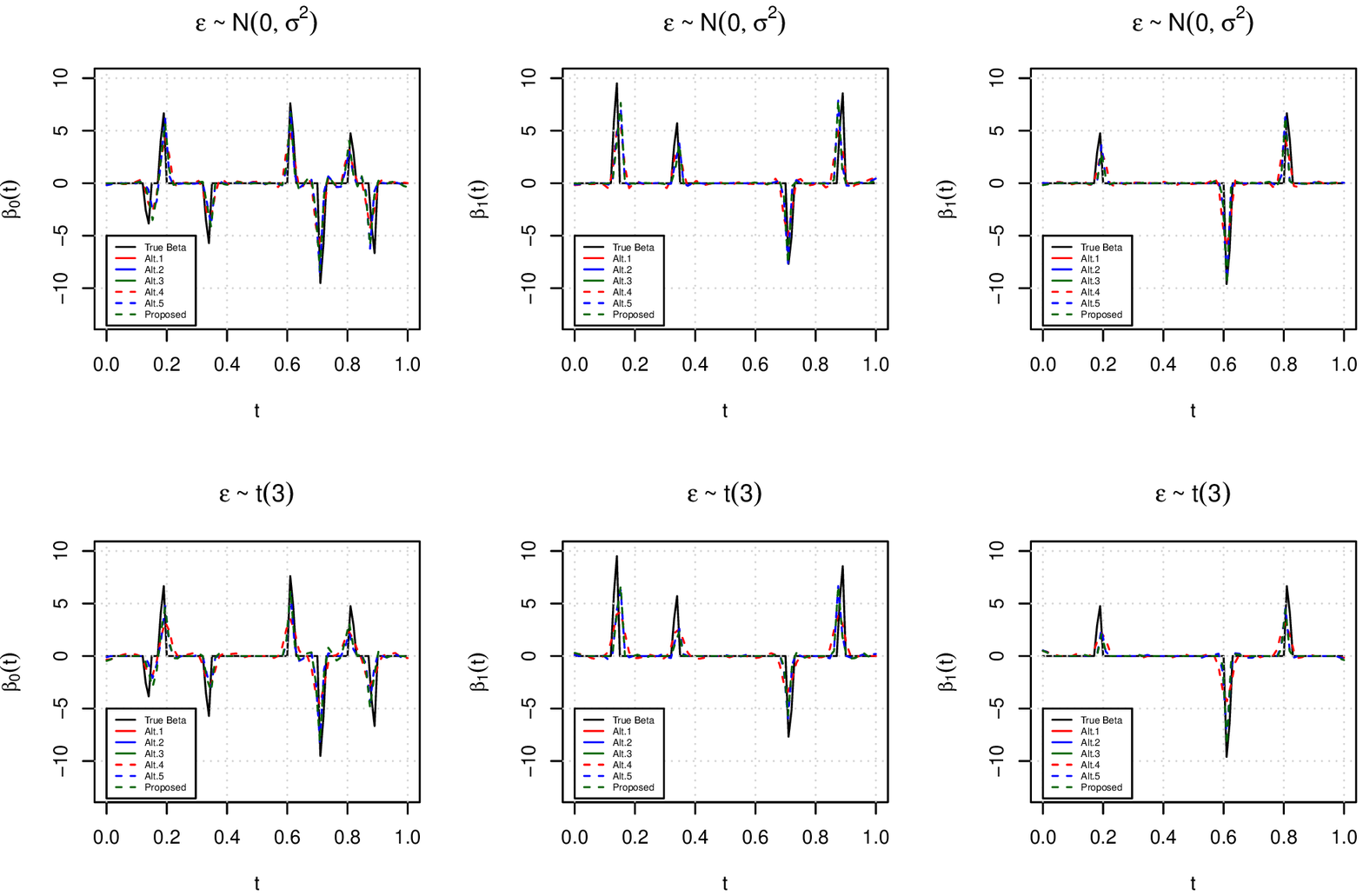}
\caption{{Average of $\hat{\beta}(t)$'s in Scenario III  with $n = 300$ based on 100 replicates under Case 1 (top) and Case 2 (bottom), respectively. Left/middle/right: $\beta_{0}(t)$/$\beta_{1}(t)$/$\beta_2(t)$.}}
\label{Fig:3}
\end{figure}

The scalar covariates ${\bm{z}}_{\cdot k}, k = 1,2$ (where ${\bm{z}}_{\cdot k}$ is the $k$-th column of $\bm{Z}$) are generated independently from the standard normal distribution, and the corresponding coefficient vector is $\bm{\gamma}^* = (0.5, 0.8)^\top$. The functional covariate $X_i(t)$ is generated as $X_i(t) = \sum a_{ij} B_j(t)$, where $a_{ij}$’s are generated from a normal distribution with mean zero and standard deviation 5, and each $B_j(t)$ is a B-spline basis function with order 5 and 71 equally spaced knots. Consider three distributions for $\epsilon_{ i}$:\\
Case 1: (homoscedasticity) $\epsilon_{ i}$ follows a normal distribution $\mathcal{N}(0, \sigma^2)$, and $\sigma$ is chosen so that the signal-to-noise ratio equals 4.\\
Case 2: (homoscedasticity) $\epsilon_{ i}$ follows a $t(3)$ distribution. \\
Case 3: (heteroscedasticity) $\epsilon_{ i} = (\frac{3}{2} |z_{i1} \int_0^1 X_i(t)\beta_1^*(t) dt |)\tilde{\epsilon}_{ i}$, where $\tilde{\epsilon}_{ i} \sim \mathcal{N}(0,1)-\mathcal{Q}_{\mathcal{N}}(\tau)$, and $\mathcal{Q}_{\mathcal{N}}(\tau)$ denotes the $\tau$th quantile of a standard norm distribution.  Note that here the model is misspecified.

\begin{table}[h]
\setlength{\abovecaptionskip}{0pt}
\setlength{\belowcaptionskip}{5pt}
\centering
\caption{Scenario I under Case 1: mean (sd) based on 100 replicates. For the quantile-based methods, $\tau = 0.5$.
}
\label{tab-1}
\resizebox{420pt}{60mm}
{
\begin{tabular}{cccccccc}
\hline
$n$   &           & Alt.1        & Alt.2     &Alt.3  & Alt.4   &Alt.5     & Proposed\\
\hline
    &           & \multicolumn{6}{c}{$\text{ISE}_0$($\times 10^{2}$)}                                    \\
300 & $\beta_0$ & 2.786(3.002) & 0.338(1.219) & 0.135(0.941) & 3.511(3.380) & 0.815(2.484) & 0.266(1.035) \\
    & $\beta_1$ & 2.524(1.514) & 0.251(0.638) & 0.279(0.861) & 2.873(1.759) & 0.407(0.926) & 0.428(1.046) \\
    & $\beta_2$ & 2.341(1.485) & 0.245(0.792) & 0.291(0.746) & 2.668(1.707) & 0.617(1.154) & 0.568(1.085) \\
500 & $\beta_0$ & 1.514(1.102) & 0.127(0.316) & 0.002(0.004) & 2.135(1.675) & 0.269(0.845) & 0.050(0.121) \\
    & $\beta_1$ & 1.533(1.050) & 0.132(0.352) & 0.043(0.228) & 1.969(1.500) & 0.254(0.668) & 0.143(0.496) \\
    & $\beta_2$ & 1.652(1.037) & 0.126(0.356) & 0.049(0.220) & 1.766(1.281) & 0.153(0.471) & 0.078(0.308) \\
    &           & \multicolumn{6}{c}{$\text{ISE}_1 (\times 10^{2})$}                                     \\
300 & $\beta_0$ & 7.742(4.976) & 7.995(5.764) & 7.542(5.389) & 9.606(6.020) & 9.969(6.746) & 9.688(6.056) \\
    & $\beta_1$ & 4.034(2.838) & 3.943(2.829) & 3.852(2.840) & 4.189(3.039) & 4.425(3.664) & 4.248(3.296) \\
    & $\beta_2$ & 3.764(2.482) & 4.212(4.063) & 4.387(3.739) & 4.580(3.134) & 5.689(5.218) & 6.092(6.107) \\
500 & $\beta_0$ & 5.347(2.526) & 5.008(2.955) & 5.412(2.948) & 6.620(3.213) & 7.254(4.033) & 7.697(3.869) \\
    & $\beta_1$ & 2.776(1.827) & 2.583(1.925) & 3.045(2.118) & 3.304(2.091) & 3.587(2.848) & 3.996(2.681) \\
    & $\beta_2$ & 2.471(1.618) & 2.363(2.010) & 2.636(2.153) & 3.016(2.118) & 3.502(2.922) & 3.910(2.907) \\
        &           & \multicolumn{6}{c}{$\text{RMSE}_{\bm{\gamma}}$}                                                 \\
300 &           & 0.054(0.044) & 0.052(0.039) & 0.051(0.040) & 0.068(0.048) & 0.068(0.053) & 0.065(0.051) \\
500 &           & 0.044(0.029) & 0.042(0.029) & 0.042(0.029) & 0.051(0.037) & 0.050(0.036) & 0.047(0.033) \\
    &           & \multicolumn{6}{c}{fTPRR}                                                                 \\
300 & $\beta_0$ & 1.000(0.000) & 0.997(0.009) & 0.999(0.004) & 1.000(0.000) & 0.997(0.012) & 0.999(0.004) \\
    & $\beta_1$ & 1.000(0.000) & 0.998(0.006) & 0.998(0.007) & 1.000(0.000) & 0.997(0.010) & 0.997(0.009) \\
    & $\beta_2$ & 1.000(0.000) & 0.996(0.020) & 0.997(0.009) & 1.000(0.000) & 0.998(0.008) & 0.994(0.018) \\
500 & $\beta_0$ & 1.000(0.000) & 0.999(0.003) & 0.999(0.003) & 1.000(0.000) & 0.999(0.004) & 0.999(0.002) \\
    & $\beta_1$ & 1.000(0.000) & 0.999(0.005) & 0.998(0.005) & 1.000(0.000) & 0.998(0.011) & 0.999(0.004) \\
    & $\beta_2$ & 1.000(0.000) & 0.999(0.006) & 0.998(0.007) & 1.000(0.000) & 0.998(0.009) & 0.999(0.004) \\
    &           & \multicolumn{6}{c}{fTNR}                                                                 \\
300 & $\beta_0$ & 0.000(0.001) & 0.756(0.205) & 0.850(0.103) & 0.000(0.001) & 0.708(0.231) & 0.793(0.156) \\
    & $\beta_1$ & 0.000(0.001) & 0.882(0.148) & 0.913(0.115) & 0.001(0.001) & 0.854(0.177) & 0.895(0.120) \\
    & $\beta_2$ & 0.001(0.001) & 0.894(0.160) & 0.922(0.089) & 0.000(0.001) & 0.847(0.175) & 0.883(0.155) \\
500 & $\beta_0$ & 0.001(0.002) & 0.750(0.174) & 0.879(0.044) & 0.001(0.001) & 0.691(0.198) & 0.763(0.152) \\
    & $\beta_1$ & 0.001(0.001) & 0.885(0.151) & 0.961(0.054) & 0.001(0.001) & 0.878(0.116) & 0.920(0.078) \\
    & $\beta_2$ & 0.001(0.001) & 0.893(0.113) & 0.958(0.044) & 0.001(0.001) & 0.894(0.100) & 0.920(0.119)\\
                     \hline
\end{tabular}
}
\end{table}

For comparison, we consider the following alternatives: (a) Alt.1 adopts the {mean squares}  lack-of-fit and a smoothness penalty (which is the last term of the proposed approach). As such, it can control smoothness as in many published studies but cannot conduct selection; 
(b) Alt.2 adopts the {mean squares} lack-of-fit and the ``functional MCP penalty + smoothness penalty’’.  {It computes conditional mean of the response and does not have a mechanism to respect the ``main effect, interaction’’ hierarchy};
(c) Alt.3 adopts the {mean squares} lack-of-fit and the same penalty as the proposed. {As such, the only difference lies in the measured conditional quantity};
(d) Alt.4 adopts the same quantile-based loss as the proposed approach and the penalty in Alt.1;
(e) Alt.5 adopts the same quantile-based loss as the proposed approach and the penalty in Alt.2. For the quantile-based approaches, we set $\tau = 0.5$ for homoscedasticity errors {(Cases 1 and 2)} and $\tau = 0.3, 0.5, 0.7$ for heteroscedasticity errors {(Case 3)}. We consider sample size $n = 300, 500$. {For each simulation replicate, we generate an independent dataset under the same setting with sample size 500 and select the optimal tunings corresponding to the best prediction}. Summary statistics are computed based on 100 independent replicates.

Performance is evaluated using the following criteria:
(a) Average integrated squared errors on null region ($\text{ISE}_0$): $\text{ISE}_{0k}= \frac{1}{l_{0k}} \int_{\mathcal{N}_k} (\hat{\beta}_k(t) - \beta_k^*(t))^2 dt,$ where $l_{0k}$ is the length of null region $\mathcal{N}_k$ of $\beta_k^*(t)$, $k = 0,1,2$.
(b) Average integrated squared errors on nonnull region ($\text{ISE}_1$):
$\text{ISE}_{1k}= \frac{1}{l_{1k}} \int_{\mathcal{N}^c_k} (\hat{\beta}_k(t) - \beta_k^*(t))^2 dt,$
where $l_{1k}$ is the length of nonnull region $\mathcal{N}^c_k$ of $\beta_k^*(t)$, $k = 0,1,2$.
(c) Root mean squared errors of $\bm{\gamma}^*$ (RMSE$_{\bm{\gamma}}$):
$\text{RMSE}_{\bm{\gamma}} =  \|\hat{\bm{\gamma}}-\bm{\gamma}^*\|_2.$
{(d) Average proportion of nonnull regions that are correctly identified (fTPR),  which is the functional counterpart of true positive rate in parametric variable selection.
(e) Average proportion of null regions that are correctly identified (fTNR),  which is the functional counterpart of true negative rate in parametric variable selection.}


{The results for Scenario I under Case 1 are provided in Table \ref{tab-1}, and those for Scenario I under Cases 2 and 3 are provided in the Appendix. In addition, the results for all cases under Scenarios II and III are provided in the supplemental materials. Figures \ref{Fig:1}-\ref{Fig:3} present the estimated $\bm{\beta}(t)$'s for Scenario I-III under Cases 1 and 2 with $n=300$.}
Overall, the findings are highly ``as expected’’. In particular, when the errors are normally distributed, the  {mean-based} methods can be advantageous. However, with Cases 2 and 3, the superiority of the quantile-based methods is obvious. In addition, it is observed that introducing local sparsity can improve estimation, and that respecting the hierarchy can further improve selection. As a representative example, consider Scenario 1 under Case 2 (Table  \ref{tab-2}, Appendix) and $n=300$. The $ \text{ISE}_{0}(\times 10^{2})$ for $\beta_2$ are 10.785,  4.577,   3.505,  6.868,  1.400, and 1.053 for the five alternative and proposed approaches, respectively. The corresponding fTNR values are 0.003 (Alt.1), 0.715 (Alt.2), 0.802 (Alt.3), 0.003 (Alt.4), 0.840 (Alt.5), and 0.891 (proposed), {and the fTPR values are similar. Furthermore, it is observed that, as the proportion of signal regions increases, it gets easier to identify sparsity.}

\section{Data analysis}
We analyze the Tecator data which is available from  http://lib.stat.cmu.edu/datasets/tecator. In this dataset, there are 215 finely chopped pure meat samples (datasets C, M, and T). For each sample, measurements are available on a spectrometric curve of spectra of absorbances measured at 100 channels with wavelength range 850-1050nm, as well as moisture, fat, and protein. The latter three are measured in percent and determined by analytic chemistry. In this analysis, we study how fat can be modeled as a function of the spectrometric curve $X(t)$ with $t$ being the wavelength and the two scalar covariates moisture $z_1$ and protein $z_2$. As developed above, we also incorporate the interactions between the spectrometric curve and scalar covariates in modeling. By introducing local sparsity, we can potentially distinguish ``useful'' regions of spectra that are informative for modeling fat from the ``noisy’’ ones. 
Prior to analysis, the range of wavelength $t$ is mapped to $[0,1]$. There are 31 equally spaced knots. Following the official guidance of this dataset, we use 129 samples (dataset C) as training for estimation, 43 samples (dataset M) for tuning parameter selection, and 43 samples (dataset T) for performance evaluation.

\begin{figure}[!h]
\centering
\includegraphics[ height=2.2 in, width=5.3 in]{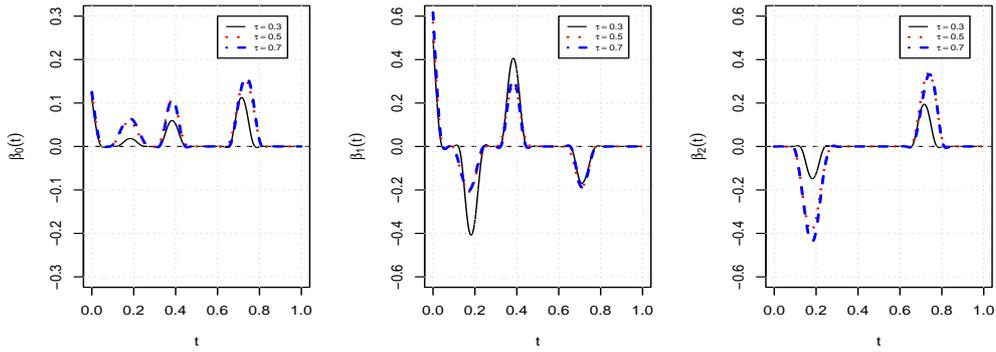}
\caption{Estimated functional effects using the proposed approach. }
\label{Fig:4}
\end{figure}


\begin{table}[h]
\setlength{\abovecaptionskip}{0pt}
\setlength{\belowcaptionskip}{5pt}
\centering
\caption{Average estimates and prediction error for the intercept and scalar covariate effects based on 100 random partitions.
}
\label{tab-6}
\resizebox{350pt}{27mm}
{
\begin{tabular}{cccccc}
\hline
\multirow{2}{*}{Method} & \multirow{2}{*}{$\tau$} & \multicolumn{3}{c}{Estimate}  & \multirow{2}{*}{Prediction Error}               \\
\cline{3-5}
&                           & $\hat{\mu}$         & $\hat{\gamma}_1$     & $\hat{\gamma}_2$     &    \\
\hline
Alt.1   &     & 5.103(0.416) & -0.070(0.011) & -0.029(0.031) & 0.039(0.006)\\
Alt.2   &   & 5.215(0.570) & -0.084(0.015) & 0.015(0.057)  & 0.037(0.006)\\
Alt.3   &  & 5.246(0.613) & -0.084(0.015) & 0.015(0.059) & 0.037(0.005)\\
Alt.4   &  0.3   & 5.375(0.424) & -0.073(0.013) & -0.035(0.040)  & 0.038(0.008)\\
   &  0.5       & 5.431(0.442) & -0.071(0.009) & -0.046(0.022) & 0.034(0.007)\\
   &  0.7      & 5.486(0.465) & -0.072(0.009) & -0.047(0.020) & 0.025(0.004) \\
Alt.5   &  0.3   & 5.720(0.493) & -0.076(0.010) & -0.047(0.025) & 0.038(0.009)\\
         & 0.5   & 5.829(0.548) & -0.076(0.008) & -0.053(0.016)  & 0.034(0.007)\\
   & 0.7   & 5.809(0.562) & -0.078(0.009) & -0.046(0.019) & 0.025(0.005) \\
Proposed  & 0.3   & 5.707(0.386) & -0.074(0.007) & -0.050(0.023) & 0.037(0.009)\\
  & 0.5  & 5.886(0.396) & -0.076(0.006) & -0.057(0.012)  & 0.033(0.007)\\
  & 0.7   & 5.908(0.386) & -0.077(0.007) & -0.053(0.014) & 0.025(0.005) \\
                       \hline
\end{tabular}
}
\end{table}

We first conduct exploratory regression analysis and present the findings in Appendix IV. Skewed residuals are observed, which justifies quantile regression. In addition, there is no obvious lack-of-fit under quantile regression. The estimation results for the functional effects are shown in Figure 4, where we consider $\tau=0.3$, 0.5, and 0.7. It is observed that the effects are locally sparse. In addition, the ``main effect, interaction’’ hierarchy is satisfied. For the scalar effects, the estimates are: $(\hat{\mu}, \hat{\gamma}_1, \hat{\gamma}_2)=(5.785, -0.068, -0.073)$ for $\tau = 0.3$, $(5.899, -0.074, -0.057)$ for $\tau = 0.5$, and $(5.930, -0.081, -0.038)$ for $\tau = 0.7$. The differences across different quantile values partly justify the need for quantile-based estimation. We recognize that a single split may not be sufficiently informative. As such, we conduct 100 random splittings of the original data, and the sizes of the three sets (under each splitting) are the same as above. In Table \ref{tab-6}, we present the mean (standard deviation) for each scalar estimate. The 100 sets of estimated functional effects are available from the authors. In addition, we also present the results of prediction error. Overall, taking the local sparsity, interpretability pertained to the variable selection hierarchy, and prediction performance into account, the analysis with the proposed approach and $\tau=0.7$ is recommended as the final one. 

\section{Discussion}
In this article, we have considered a more sophisticated functional data analysis model. The most significant advancement comes from the interaction analysis. A new estimation and variable selection method has been developed, and its theoretical and numerical properties have been carefully investigated. The proposed model can be potentially extended to include more complicated interactions (for example, between functional effects) and have higher dimensions. It will also be of interest to examine more practical applications. 

\section*{Acknowledgements}
We thank the associate editor and reviewers for careful review and insightful comments. 
This study has been partly supported by the National Natural Science
Foundation of China [11971404], National Bureau of Statistics of China [2022LZ34], Fundamental Research Funds for the Central Universities,  Research Funds of Renmin University of China [21XNH152], and NIH [CA204120].

\clearpage

\bibliography{Citation}

\begin{thebibliography}{10}
\expandafter\ifx\csname url\endcsname\relax
  \def\url#1{\texttt{#1}}\fi
\expandafter\ifx\csname urlprefix\endcsname\relax\def\urlprefix{URL }\fi
\expandafter\ifx\csname href\endcsname\relax
  \def\href#1#2{#2} \def\path#1{#1}\fi

\bibitem{aneiros2021variable}
G.~Aneiros, S.~Novo, P.~Vieu, Variable selection in functional regression
  models: A review, Journal of Multivariate Analysis (2021) 104871.

\bibitem{cardot2003spline}
H.~Cardot, F.~Ferraty, P.~Sarda, Spline estimators for the functional linear
  model, Statistica Sinica (2003) 571--591.

\bibitem{fan2013functional}
R.~Fan, Y.~Wang, J.~L. Mills, A.~F. Wilson, J.~E. Bailey-Wilson, M.~Xiong,
  Functional linear models for association analysis of quantitative traits,
  Genetic epidemiology 37~(7) (2013) 726--742.

\bibitem{yao2005functional}
F.~Yao, H.-G. M{\"u}ller, J.-L. Wang, Functional linear regression analysis for
  longitudinal data, The Annals of Statistics 33~(6) (2005) 2873--2903.

\bibitem{berrendero2019rkhs}
J.~R. Berrendero, B.~Bueno-Larraz, A.~Cuevas, An rkhs model for variable
  selection in functional linear regression, Journal of Multivariate Analysis
  170 (2019) 25--45.

\bibitem{shin2016rkhs}
H.~Shin, S.~Lee, An rkhs approach to robust functional linear regression,
  Statistica Sinica (2016) 255--272.

\bibitem{tong2018analysis}
H.~Tong, M.~Ng, Analysis of regularized least squares for functional linear
  regression model, Journal of Complexity 49 (2018) 85--94.

\bibitem{cui2017estimation}
X.~Cui, Y.~Lu, H.~Peng, Estimation of partially linear regression models under
  the partial consistency property, Computational Statistics \& Data Analysis
  115 (2017) 103--121.

\bibitem{fan2003kernel}
Y.~Fan, Q.~Li, A kernel-based method for estimating additive partially linear
  models, Statistica Sinica (2003) 739--762.

\bibitem{wu2014integrative}
C.~Wu, Y.~Cui, S.~Ma, Integrative analysis of gene--environment interactions
  under a multi-response partially linear varying coefficient model, Statistics
  in medicine 33~(28) (2014) 4988--4998.

\bibitem{james2009functional}
G.~M. James, J.~Wang, J.~Zhu, Functional linear regression that’s
  interpretable, The Annals of Statistics 37~(5A) (2009) 2083--2108.

\bibitem{lin2017locally}
Z.~Lin, J.~Cao, L.~Wang, H.~Wang, Locally sparse estimator for functional
  linear regression models, Journal of Computational and Graphical Statistics
  26~(2) (2017) 306--318.

\bibitem{zhou2013functional}
J.~Zhou, N.-Y. Wang, N.~Wang, Functional linear model with zero-value
  coefficient function at sub-regions, Statistica Sinica 23~(1) (2013) 25.

\bibitem{kong2016partially}
D.~Kong, K.~Xue, F.~Yao, H.~H. Zhang, Partially functional linear regression in
  high dimensions, Biometrika 103~(1) (2016) 147--159.

\bibitem{ma2019quantile}
H.~Ma, T.~Li, H.~Zhu, Z.~Zhu, Quantile regression for functional partially
  linear model in ultra-high dimensions, Computational Statistics \& Data
  Analysis 129 (2019) 135--147.

\bibitem{yao2017regularized}
F.~Yao, S.~Sue-Chee, F.~Wang, Regularized partially functional quantile
  regression, Journal of Multivariate Analysis 156 (2017) 39--56.

\bibitem{zhang2010nearly}
C.~Zhang, Nearly unbiased variable selection under minimax concave penalty, The
  Annals of statistics 38~(2) (2010) 894--942.

\bibitem{liu2013sparse}
J.~Liu, J.~Huang, Y.~Xie, S.~Ma, Sparse group penalized integrative analysis of
  multiple cancer prognosis datasets, Genetics research 95~(2-3) (2013) 68--77.

\bibitem{de1978practical}
C.~De~Boor, C.~De~Boor, A practical guide to splines, Vol.~27, springer-verlag
  New York, 1978.

\bibitem{shi2014penalized}
X.~Shi, J.~Liu, J.~Huang, Y.~Zhou, Y.~Xie, S.~Ma, A penalized robust method for
  identifying gene--environment interactions, Genetic epidemiology 38~(3)
  (2014) 220--230.

\bibitem{wu2020structured}
M.~Wu, Q.~Zhang, S.~Ma, Structured gene-environment interaction analysis,
  Biometrics 76~(1) (2020) 23--35.

\bibitem{schumaker2007spline}
L.~Schumaker, Spline functions: basic theory, Cambridge University Press, 2007.

\bibitem{sherwood2016partially}
B.~Sherwood, L.~Wang, Partially linear additive quantile regression in
  ultra-high dimension, The Annals of Statistics 44~(1) (2016) 288--317.

\bibitem{stone1985additive}
C.~J. Stone, Additive regression and other nonparametric models, The annals of
  Statistics 13~(2) (1985) 689--705.

\end{thebibliography}

\section*{Appendix}

\renewcommand{\theequation}{A.\arabic{equation}}
\setcounter{equation}{0}

\subsection*{I. Lemma 1 and remarks}

\begin{lemma}
\label{lemma:1} 
Consider $M_n+1$ equally spaced knots $0=t_0<t_1<\cdots <t_{M_n} = T$ in the domain $[0,T]$. For the smooth functional main effect and interactions, we have:
\begin{equation*}
\label{eq:1.2}
\begin{aligned}
&\sum_{k = 1}^q \frac{1}{T}\int_0^T p_{\lambda_1} (|\beta_{k}(t)|)dt + \frac{1}{T} \int_0^T p_{\lambda_2}(\|\bm{\beta}(t)\|_2)dt\\
&= \lim_{M_n \rightarrow \infty} \frac{1}{M_n} \sum_{k=1}^q  \sum_{l = 1}^{M_n} p_{\lambda_1} \left( {{M_n}^{\frac{1}{2}}{T}^{-\frac{1}{2}} \|\beta_{k[l]} \|  } \right) +  \lim_{M_n \rightarrow \infty}    \frac{1}{M_n} \sum_{l = 1}^{M_n} p_{\lambda_2} \left( {M_n}^{\frac{1}{2}}{T}^{-\frac{1}{2}} \|\bm{\beta}_{[l]}\| \right),
\end{aligned}
\end{equation*}
where $\|\beta_{k[l]} \| = ( \int_{t_{l-1}}^{t_l} \beta_{k}^2(t)dt )^{1/2}$ and $\|\bm{\beta}_{[l]}\| = ( \sum_{k = 0}^q \int_{t_{l-1}}^{t_l} \beta_{k}^2(t) dt  )^{1/2}$.
\end{lemma}

This lemma can be derived from Theorem 1 of \cite{lin2017locally}. It shows that the penalty evaluated over the whole domain is asymptotically equivalent to the sum over a large number of subregions. This nicely matches the spline basis expansion framework. For each subregion, we note that the penalty still has a sparse group form. As such, the ``main effect, interaction’’ hierarchy is expected to hold for each subregion (and so the whole domain).

\clearpage
\subsection*{II: Proof of Theorem 1}

Let $C$ be a generic positive constant which may take different values under different circumstances. Denote: 
$$g^*(X_i(t),\bm{z}_i)=\int_0^T X_i(t)\beta_0^*(t)dt + \sum_{k=1}^{q} z_{ik} \int_0^T X_i(t)\beta_k^*(t)dt.$$
Recall that $\mathop{\rm{\text{sup}}}_{t \in [0,T]} \left|\beta_k^*(t) - \bm{B}^\top(t) \bm{b}_k^*\right| = O(M_n^{-r})$. With the boundedness Condition 3, we have  $g^*(X_i(t), \bm{z}_i) = \bm{\psi}^\top_i \bm{b}^* + O( M_n^{-r})$. Denote the empirical version of the projection of $\bm{z}_{\cdot k}$ onto the spline approximation of the functional covariate space as $\bm{h}_{\cdot k} = \bm{\Psi} \hat{\bm{\varpi}}_k$, where $\bm{z}_{\cdot k}$ is the $k$th column of $\bm{Z}$ and  $\hat{\bm{\varpi}}_k$ is the minimizer of: 
\begin{equation*}
\min_{{\bm{\varpi}}_k \in \mathcal{R}^{q_n}} \sum_{i = 1}^n f_i(0)(z_{ik} - \bm{\psi}_i^\top \bm{\varpi}_k)^2.
\end{equation*}
The solution to {the above problem} is $\hat{\bm{\varpi}}_k = (\bm{\Psi}^\top \bm{B}_n \bm{\Psi})^{-1}\bm{\Psi}^\top \bm{B}_n \bm{z}_{\cdot k}$. Let $\bm{H}$ be the $n\times q$ matrix with the $k$th column being $\bm{h}_{\cdot k}$. We define the projection matrix $\bm{P} = \bm{\Psi}(\bm{\Psi}^\top \bm{B}_n \bm{\Psi})^{-1}\bm{\Psi}^\top \bm{B}_n \in \mathcal{R}^{n \times n}$, and it is obvious that $\bm{H} = \bm{P}\bm{Z}$.  
Thus we have $\check{\bm{Z}} = (\check{\bm{z}}_1, \cdots, \check{\bm{z}}_n)^\top = (\bm{I}_n - \bm{P}) \bm{Z}$. 

Define $\tilde{\bm{z}}_i  = n^{-\frac{1}{2}}\check{\bm{z}}_i \in \mathcal{R}^q$, $\bm{\Psi}_B^2 = \bm{\Psi}^\top \bm{B}_n \bm{\Psi} \in \mathcal{R}^{q_n\times q_n}$, and $\tilde{\bm{\psi}}_{i} = \bm{\Psi}_B^{-1} \bm{\psi}_{i} \in \mathcal{R}^{q_n}$. 
Following \cite{sherwood2016partially}, we reparameterize the quantile loss function as:
\begin{equation*}
\begin{aligned}
 \rho_\tau \left(y_i - \bm{\psi}_i^\top \bm{b} - \bm{z}_i^\top \bm{\gamma} \right) 
= \rho_\tau \left(\epsilon_i -\tilde{\bm{z}}_i^\top \bm{\theta}_1 - \tilde{\bm{\psi}}_i^\top \bm{\theta}_2 - u_{ni} \right),
\end{aligned}
\end{equation*}
where $\bm{\theta}_1 = \sqrt{n}(\bm{\gamma}- \bm{\gamma}^*) \in \mathcal{R}^q$, $\bm{\theta}_2 = \bm{\Psi}_B (\bm{b} - \bm{b}^*) +\bm{\Psi}_B^{-1}\bm{\Psi}^\top \bm{B}_n \bm{Z} (\bm{\gamma}- \bm{\gamma}^*) \in \mathcal{R}^{q_n}$ and $u_{ni} = \bm{\psi}_{i}^\top \bm{b}^* -g^*(X_i(t), \bm{z}_i)$. Let $\bm{\theta} = (\bm{\theta}_1^\top, \bm{\theta}_2^\top)^\top$. The objective function under the reparameterization is:
\begin{equation*}
\begin{aligned}
\tilde{Q}(\bm{\theta})  = & \frac{1}{n} \sum_{i=1}^n \rho_\tau (\epsilon_i -\tilde{\bm{z}}_i^\top \bm{\theta}_1 - \tilde{\bm{\psi}}_i^\top \bm{\theta}_2 - u_{ni} ) + \sum_{k = 1}^q \sum_{l = 1}^{M_n} p_{\lambda_1} (\| \bm{b}_k \|_{\bm{W}_l})\\
& + \sum_{l = 1}^{M_n} p_{\lambda_2} (\| \bm{b} \|_{\bm{W}_l}) + \eta \sum_{k=0}^q  \bm{b}_k^\top \bm{V}\bm{b}_k.
\end{aligned}
\end{equation*}

Define: 
\begin{equation*}
\begin{aligned}
D_i(\bm{\theta}) = & \rho_\tau ( \epsilon_i - \tilde{\bm{z}}_i^\top \bm{\theta}_1 - \tilde{\bm{\psi}}_i^\top \bm{\theta}_2 - u_{ni}) -  \rho_\tau ( \epsilon_i - u_{ni}) + (\tilde{\bm{z}}_i^\top \bm{\theta}_1 + \tilde{\bm{\psi}}_i^\top \bm{\theta}_2) \mathcal{D}_\tau(\epsilon_i) \\
& - E[\rho_\tau ( \epsilon_i - \tilde{\bm{z}}_i^\top \bm{\theta}_1 - \tilde{\bm{\psi}}_i^\top \bm{\theta}_2 - u_{ni}) - \rho_\tau ( \epsilon_i - u_{ni}) ], 
\end{aligned}
\end{equation*}
where $\mathcal{D}_\tau(\epsilon_i) = \tau - I(\epsilon_i < 0)$. 
We first state the following lemmas.
\begin{lemma}
\label{lemma2}
Let $d_n = q_n + q$. Under Conditions 1-4,  for any positive constant $L$, we have:
\begin{equation*}
\sup_{\|\bm{\theta} \|_2 \leq L\sqrt{d_n}} \frac{1}{{d_n}} \left| \sum_{i=1}^n D_i(\bm{\theta}) \right| = o_p(1).
\end{equation*}
\end{lemma}
\begin{proof}
Proof follows that of Lemma B.1 in \cite{sherwood2016partially} under Conditions 1-4.
\end{proof}

\begin{lemma}
\label{lemma3}
Let $\tilde{\bm\theta}_1   = \sqrt{n} \left( \check{\bm{Z}}^{\top} \bm{B}_n \check{\bm{Z}}\right)^{-1}\check{\bm{Z}} \mathcal{D}_\tau (\bm{\epsilon})$, where $\mathcal{D}_\tau (\bm{\epsilon}) = \left(\mathcal{D}_\tau (\epsilon_1), \cdots, \mathcal{D}_\tau (\epsilon_n)\right)^\top$. Under Conditions 1-4, we have $\| \tilde{\bm\theta}_1 \|_2 = O_p(1)$.
\end{lemma}

\begin{proof}
Proof follows from that of Lemma 5 (1) in \cite{sherwood2016partially}.
\end{proof}

\noindent\emph{Proof of Theorem 1 (1)} 
Here we show that there exists a local minimizer $\hat{\bm{\theta}}= (\hat{\bm{\theta}}_1^{\top}, \hat{\bm{\theta}}_2^{\top})^{\top}$ of (\ref{eq:6}) such that $\| {\hat{\bm{\theta}} }\|_2 = O_p(\sqrt{M_n})$ and $\| \hat{\bm{\theta}}_1 \|_2 = O_p(1)$.


Note that $d_n= O(M_n)$.  To prove $\| {\hat{\bm{\theta}} }\|_2 = O_p(\sqrt{M_n})$, it is sufficient to show that, for any $\delta > 0$, there exists a sufficiently large positive constant $L$ such that:
\begin{equation}
\label{eq:A.1}
P\left\{ \inf_{\|\bm{\theta}\|_2 \leq L\sqrt{d_n}} \tilde{Q}(\bm{\theta}) > \tilde{Q}(\bm{0})    \right\} \geq 1-\delta.
\end{equation}
That is, with probability at least $1-\delta$, there exists a local minimizer such that $\|\hat{\bm{\theta}}\|_2 \leq L\sqrt{d_n}$.

We first show that, for a sufficiently large positive $L$, there exists a positive constant $C$ such that:
\begin{equation}
\label{eq:A.2}
 \inf_{\|\bm{\theta} \|_2  = L\sqrt{d_n}} \frac{1}{n} \sum_{i= 1}^n \left[ \rho_\tau(\epsilon_i - \tilde{\bm{z}}_i^\top \bm{\theta}_1  - \tilde{\bm{\psi}}_i^\top \bm{\theta}_2 -u_{ni}) -  \rho_\tau (\epsilon_i  -u_{ni} )  \right] > CL^2 {d_n}/n
\end{equation}
with probability tending to one.
From Lemma \ref{lemma2}, we have:
\begin{equation*}
\begin{aligned}
&\sup_{\|\bm{\theta} \|_2 \leq L\sqrt{d_n}} \frac{1}{{n}} \left| \sum_{i=1}^n D_i(\bm{\theta}) \right| \\
&=  \sup_{\|\bm{\theta} \|_2 \leq L\sqrt{d_n}} \frac{1}{{n}}  \left|   \sum_{i=1}^n \rho_\tau (\epsilon_i - \tilde{\bm{z}}_i^\top \bm{\theta}_1  - \tilde{\bm{\psi}}_i^\top \bm{\theta}_2 -u_{ni} )  -\sum_{i=1}^n \rho_\tau (\epsilon_i  -u_{ni} ) \right.\\
&\left. \quad \quad \quad \quad \quad \quad \quad \  - \sum_{i=1}^n E\left[\rho_\tau (\epsilon_i - \tilde{\bm{z}}_i^\top \bm{\theta}_1  - \tilde{\bm{\psi}}_i^\top \bm{\theta}_2 -u_{ni} ) - \rho_\tau (\epsilon_i  -u_{ni} )\right] \right.\\
& \left. \quad \quad \quad \quad \quad \quad \quad \ + \sum_{i=1}^n (\tilde{\bm{z}}_i^\top \bm{\theta}_1  + \tilde{\bm{\psi}}_i^\top \bm{\theta}_2)\mathcal{D}_\tau (\epsilon_i)  \right| =o_p(d_n/n).
\end{aligned}
\end{equation*}
Denote $F_{n1} = n^{-1}\sum_{i=1}^n  E [ \rho_\tau (\epsilon_i - \tilde{\bm{z}}_i^\top \bm{\theta}_1  - \tilde{\bm{\psi}}_i^\top \bm{\theta}_2 -u_{ni} ) -  \rho_\tau (\epsilon_i  -u_{ni} ) ]$ and $F_{n2} = n^{-1}  \sum_{i=1}^n (\tilde{\bm{z}}_i^\top \bm{\theta}_1  + \tilde{\bm{\psi}}_i^\top \bm{\theta}_2)\mathcal{D}_\tau (\epsilon_i) $. Following similar arguments as in the proof of Lemma 4 in \cite{sherwood2016partially}, we can show that for a sufficiently large positive $L$, $F_{n1}$ has asymptotically a lower bound of $CL^2 d_n/n$ and $F_{n2} =O_p(d_n^{1/2}/n)$. Therefore, (\ref{eq:A.2}) is proved. 

Let $D$ denote the domain $[0,T]$. For given $\lambda_1, \lambda_2$ and $M_n$, and for each $\beta_k^*(t)$, we divide $D$ into three parts: the first part $D_k^{[1]} = \{t \in D: |\beta_k^*(t)| \geq C\xi \max(\lambda_1, \lambda_2)\}$ for some constant $C > 1$, the second part $D_k^{[2]} = \{ t \in D: \beta_k^*(t) = 0 \}$, and the third part $D_k^{[3]} = \{ t \in D: 0 < |\beta_k^*(t)| <  C\xi  \max(\lambda_1, \lambda_2)  \}$. Since $\max (\lambda_1, \lambda_2) \rightarrow 0$ as $n\rightarrow \infty$, $D_k^{[3]}$ shrinks to the empty set $\emptyset$ as $n\rightarrow \infty$.

Next, we consider the penalty terms. With $\|\bm{\theta}\|_2 = O(\sqrt{d_n})$ and the definition of $\bm{\theta}$, we have $\|\bm{\gamma}-\bm{\gamma}_0\|_2=O(\sqrt{d_n/n})$. In addition, 
\begin{equation}\label{eq:A.3}
\|\bm{\Psi}_B (\bm{b}- \bm{b}^*)\|_2^2 \leq 2\|\bm{\theta}_2\|_2^2 + 2\|\bm{\Psi}_B^{-1}\bm{\Psi}^\top \bm{B}_n \bm{Z}(\bm{\gamma}-\bm{\gamma}_0)\|_2^2 = O(d_n).
\end{equation}
The last equality holds because $\|\bm{\Psi}_B^{-1}\bm{\Psi}^\top \bm{B}_n \bm{Z}(\bm{\gamma}-\bm{\gamma}_0)\|_2^2 = O(n \|\bm{\gamma}-\bm{\gamma}_0\|_2^2)$ by Conditions 1 and 3. Then we have $\| \bm{b}_k - \bm{b}_k^* \|_2 = O(d_n n^{-1/2})$.
Notice that $\beta_k^*(t) = \bm{B}^\top(t) \bm{b}_k^* + O(M_n^{-r})$. For a subregion $I_l \subset D_k^{[1]}, k = 0, \cdots, q$, with $M_n=O(d_n)$, Condition 4, and $n^{-\frac{r}{2r+1}}/\min(\lambda_1, \lambda_2) = o(1)$, we have: 
$$
\| \bm{b}_k^* \|_{\bm{W}_l} = \sqrt{\frac{M_n}{T} \int_{t_{l-1}}^{t_l} \beta_{k}^{*2}(t)dt } + O(M_n^{-r}) \geq C\xi \max(\lambda_1, \lambda_2).
$$
Applying some inequality techniques, we can derive $\|\bm{b}_k\|_{\bm{W}_l} \geq C\xi \max (\lambda_1, \lambda_2)$. With the properties of MCP and $C>1$, we have $p_{\lambda_1} ( \|\bm{b}_k \|_{\bm{W}_l} ) = p_{\lambda_1} (\|\bm{b}_{k}^* \|_{\bm{W}_l} )$ and $p_{\lambda_2} ( \|\bm{b} \|_{\bm{W}_l} ) = p_{\lambda_2} (\|\bm{b}^* \|_{\bm{W}_l} )$ for $l$ satisfying $I_l \subset D_k^{[1]}$. 
For a subregion $I_l \subset D_k^{[2]} \cap (\cup_{k^\prime \neq k} D_{k^\prime}^{[1]})$, 
by the choice of $\bm{b}^*$, we have $\| \bm{b}_k^*\|_{\bm{W}_l} = 0$ and $\|\bm{b}^*\|_{\bm{W}_l} \geq C \xi\max(\lambda_1, \lambda_2)$, and thus $p_{\lambda_1}(\| \bm{b}_k \|_{\bm{W}_l}) \geq p_{\lambda_1}(\| \bm{b}_k^* \|_{\bm{W}_l})=0$ and $p_{\lambda_2} ( \|\bm{b} \|_{\bm{W}_l} ) = p_{\lambda_2} (\|\bm{b}^* \|_{\bm{W}_l} )$.
For a subregion $I_l \subset D_k^{[2]} \cap (\cup_{k^\prime \neq k} D_{k^\prime}^{[1]})^c$, we have $\|\bm{b}^*\|_{\bm{W}_l} =0$, and thus $p_{\lambda_1}(\| \bm{b}_k \|_{\bm{W}_l}) \geq p_{\lambda_1}(\| \bm{b}_k^* \|_{\bm{W}_l})=0$ and $p_{\lambda_2} ( \|\bm{b} \|_{\bm{W}_l} ) \geq p_{\lambda_2} (\|\bm{b}^* \|_{\bm{W}_l} )=0$. Summarizing the above three cases, we have:
\begin{equation}
\label{eq:A.4}
\sum_{k=1}^q \sum_{l = 1}^{M_n} p_{\lambda_1} \left(  \|\bm{b}_k \|_{\bm{W}_l} \right) \geq  \sum_{k=1}^q \sum_{l = 1}^{M_n}p_{\lambda_1} \left(   \|\bm{b}_{k}^* \|_{\bm{W}_l} \right) 
\end{equation}
and
\begin{equation}
\label{eq:A.5}
 \sum_{l = 1}^{M_n} p_{\lambda_1} \left(   \|\bm{b} \|_{\bm{W}_l}\right) \geq  \sum_{l = 1}^{M_n}p_{\lambda_1} \left(  \|\bm{b}^* \|_{\bm{W}_l}\right). 
\end{equation}

Also, by the Cauchy-Schwarz inequality and $\eta = o(n^{-1/2})$, we have:
\begin{eqnarray}
\label{eq:A.6}
\sum_{k=0}^q \eta \bm{b}_k^\top \bm{V} \bm{b}_k -\sum_{k=0}^q \eta \bm{b}_{k}^{*\top} \bm{V} \bm{b}_{k}^* & =&
\sum_{k = 0}^q \eta \left[(\bm{b}_k-\bm{b}_{k}^*)^\top \bm{V} (\bm{b}_k-\bm{b}_{k}^*) + 2(\bm{b}_k-\bm{b}_{k}^*)^\top \bm{V} \bm{b}_k^* \right] \nonumber\\
& \leq & O(\eta d_n n^{-1}  +\eta n^{-1/2})= o( n^{-1}),
\end{eqnarray}
where the inequality follows from the fact that $\| \bm{b}_k - \bm{b}_k^* \|_2 = O(d_n n^{-1/2})$, $\lambda_{\max}(\bm{V}) = O(d_n^{-1})$ and $\sup_j |{\bm{V}}_{\cdot j} \bm{b}_k^*| \leq C d_n^{-1}$, where ${\bm{V}}_{j\cdot}$ is the $j$th row of $\bm{V}$ for $j = 1,\cdots, M_n+d$.

Combining (\ref{eq:A.2}), (\ref{eq:A.4}), (\ref{eq:A.5}) and (\ref{eq:A.6}), for $\|\bm{\theta} \|_2 = L\sqrt{d_n}$ and a sufficiently large $L$, we prove (\ref{eq:A.1}).
Therefore, there exists a local minimizer $\hat{\bm{\theta}}$ such that $\|\hat{\bm{\theta}} \|_2 = O_p(\sqrt{d_n})$. 
Similar to (\ref{eq:A.3}), it follows that $\| \bm{\Psi}_B (\hat{\bm{b}} - \bm{b}^*) \|_2 = O_p(\sqrt{d_n} )$, and thus $\| \hat{\bm{b}}- \bm{b}^* \|_2 = O_p(d_n n^{-1/2})$. Then we have:
\begin{eqnarray*}
\int_0^T (\hat{\beta}_{k}(t) - \beta_k^*(t) )^2 dt
&\leq & 2\int_0^T (\hat{\beta}_{k}(t) - \bm{B}^\top (t)\bm{b}_{k}^*)^2dt  + 2\int_0^T (\bm{B}^{\top}(t) \bm{b}_{k}^* -\beta_k^*(t) )^2dt\\
& = & O(d_n^{-1}\| \hat{\bm{b}}- \bm{b}^* \|_2^2)  + O_p(M_n^{-2r}) =O_p(n^{-\frac{2r}{2r+1}}),
\end{eqnarray*}
where the first inequality follows from the triangle inequality, and the last equality is due to $d_n= O(M_n)$ and Condition 4.



Next, we examine the convergence rate of $\hat{\bm{\theta}}_1$. To verify $\| \hat{\bm{\theta}}_1\|_2 = O_p(1)$, it is sufficient to show that $\| \hat{\bm{\theta}}_1 - \tilde{\bm{\theta}}_1 \|_2 = o_p(1)$ under Conditions 1-4. Define:
\begin{equation*}
\tilde{Q}_i (\bm{\theta}_1, \tilde{\bm{\theta}}_1, \bm{\theta}_2) = \rho_\tau (\epsilon_i - \tilde{\bm{z}}_i^\top \bm{\theta}_1 - \tilde{\bm{\psi}}_2^\top \bm{\theta}_2 - u_{ni}) - \rho_\tau (\epsilon_i - \tilde{\bm{z}}_i^\top \tilde{\bm{\theta}}_1 - \tilde{\bm{\psi}}_2^\top \bm{\theta}_2 - u_{ni}).
\end{equation*}

We first show that for any positive constants $M$ and $C$,
\begin{equation}
\label{eq:A.7}
P\left(\inf_{\substack{\| \bm{\theta}_1 - \tilde{\bm{\theta}}_1 \|_2 \geq M \\ \|\bm{\theta}_2\|_2 \leq C\sqrt{d_n}} } \sum_{i = 1}^n \tilde{Q}_i (\bm{\theta}_1, \tilde{\bm{\theta}}_1, \bm{\theta}_2) > 0\right)
\rightarrow 1.
\end{equation}
Following the proof of Lemma 6 in \cite{sherwood2016partially}, we have:
\begin{equation*}
\sup_{\substack{\| \bm{\theta}_1 - \tilde{\bm{\theta}}_1 \|_2 \leq M \\ \|\bm{\theta}_2\|_2 \leq C\sqrt{d_n}} } \left| \sum_{i = 1}^n \tilde{Q}_i (\bm{\theta}_1, \tilde{\bm{\theta}}_1, \bm{\theta}_2) - \frac{1}{2} (\bm{\theta}_1 - \tilde{\bm{\theta}}_1)^\top \left(\frac{1}{n} \check{\bm{Z}}^{\top} \bm{B}_n \check{\bm{Z}}\right)(\bm{\theta}_1 - \tilde{\bm{\theta}}_1) (1+o_p(1)) \right|= o_p(1).
\end{equation*}
By Conditions 1 and 3, for any $\|\bm{\theta}_1 -\tilde{\bm{\theta}}_1 \|_2 > M$, 
$$
\frac{1}{2} (\bm{\theta}_1 - \tilde{\bm{\theta}}_1)^\top \left(\frac{1}{n} \check{\bm{Z}}^{\top} \bm{B}_n  \check{\bm{Z}} \right)(\bm{\theta}_1 - \tilde{\bm{\theta}}_1) > CM,
$$
for some positive constant C, and thus (\ref{eq:A.7}) holds.
Combining (\ref{eq:A.1}), (\ref{eq:A.7}) and Lemma \ref{lemma3}, we have that there exists a local minimizer $\hat{\bm{\theta}}_1$ of ({\ref{eq:6}}) such that $\|\hat{\bm{\theta}}_1\|_2  = O_p(1)$, and thus $\| \hat{\bm{\gamma}} - \bm{\gamma}^* \|_2 = O_p(\sqrt{1/n})$.

\noindent \emph{Proof of Theorem 1 (2)} 

We need to show that $\hat{\beta}_k(t) = 0$ for all $t \in D_k^{[2]}$ with probability tending to one. Denote $\hat{\bm{b}}_k^{[l]} = (\hat{b}_{k,l}, \cdots, \hat{b}_{k, l+ d})^\top$. We need to prove that the local minimizer $ (\hat{\bm{b}}^\top, \hat{\bm{\gamma}}^\top)^\top$ satisfies $\|\hat{\bm{b}}_k^{[l]}\|_2 = 0$ for all $l$ such that $I_l \subset D_k^{[2]}$ with probability tending to one for $k = 0, \cdots, q$. 
By the way of contradiction, assume that $\| \hat{\bm{b}}_{k}^{[l^\star]}\|_2\neq 0$ for some $l^{\star}$ with $I_{l^{\star}} \subset D_k^{[2]}$. 
Let $\tilde{\bm{b}}_k$ be the same as $\hat{\bm{b}}_k$ except that $\|\tilde{\bm{b}}_{k}^{[l^\star]}\|_2 =0$. Note that $\| \tilde{\bm{b}}_k^{[l]}\|_2 = 0$ is equivalent to $\| \tilde{\bm{b}}_{k}\|_{\bm{W}_l} = 0$ for $l = 1,\cdots, M_n$. 
Since $\|\bm{b}_{k}^{*[l^\star]}\|_2 = 0$ and $\| \bm{b}_{k}^* - \hat{\bm{b}}_k \|_2 = O_p(M_n/\sqrt{n}) $, we have $\|\hat{\bm{b}}_k^{[l^\star]} \|_2 = O(M_n/\sqrt{n})$, and thus $\|\hat{\bm{b}}_k \|_{\bm{W}_{l^\star}} = O(\sqrt{M_n/n})$ by $\lambda_{\max}{(\bm{W}_{l^{\star}})} = O(M_n^{-1})$. Below we prove that:
\begin{equation}
\label{eq:A.8}
\begin{aligned}
 &\frac{1}{n} \sum_{i= 1}^n  \rho_\tau(y_i - {\bm{\psi}}_i^\top \hat{\bm{b}}- {\bm{z}}_i^\top \hat{\bm{\gamma}} )    + \sum_{k=1}^q \sum_{l = 1}^{M_n} p_{\lambda_1} ( \|\hat{\bm{b}}_k\|_{\bm{W}_l}   ) +  \sum_{l = 1}^{M_n} p_{\lambda_2} ( \|\hat{\bm{b}}\|_{\bm{W}_l} ) + \eta \sum_{k=0}^q  \hat{ \bm{b}}_k^\top \bm{V} \hat{\bm{b}}_k \\
&>  \frac{1}{n} \sum_{i= 1}^n  \rho_\tau(y_i - {\bm{\psi}}_i^\top \tilde{\bm{b}}- {\bm{z}}_i^\top \hat{\bm{\gamma}} )    + \sum_{k=1}^q \sum_{l = 1}^{M_n} p_{\lambda_1} (  \|\tilde{\bm{b}}_k\|_{\bm{W}_l} )   +  \sum_{l = 1}^{M_n} p_{\lambda_2} (   \|\tilde{\bm{b}}\|_{\bm{W}_l}  )+ \eta \sum_{k=0}^q  \tilde{ \bm{b}}_k^\top \bm{V} \tilde{\bm{b}}_k ,
 \end{aligned}
 \end{equation}
with probability tending to one, and this leads to a contradiction. 
Therefore, we conclude that $\|\hat{\bm{b}}_k^{[l]}\|_2 = 0$ for all $l\subset D_k^{[2]}$ in probability. Furthermore, by the definition of $\hat{\beta}_k(t)$, we have $\hat{\beta}_k(t) = 0$ for all $t \in D_k^{[2]}$ with probability tending to one. 

By the convexity of the {quantile} loss function, we have
\begin{eqnarray}
\label{eq:A.9}
&&\frac{1}{n} \sum_{i = 1}^n (  \rho_\tau(y_i - \bm{\psi}_i^\top \hat{\bm{b}} - \bm{z}_i^\top \hat{\bm{\gamma}}) -   \rho_\tau(y_i - \bm{\psi}_i^\top \tilde{\bm{b}} - \bm{z}_i^\top \hat{\bm{\gamma}})) \nonumber\\
&\geq& - \frac{1}{n} \sum_{i = 1}^n ( \tau-1{(y_i \leq \bm{\psi}_i^\top \tilde{\bm{b}} + \bm{z}_i^\top \hat{\bm{\gamma}})} ) \bm{\psi}_{i}^{ [kl^{\star}]\top} \hat{\bm{b}}_{k}^{[l^{\star}]}\\
&=& -\frac{1}{n}  \sum_{i=1}^n \left(\tau-1{(\epsilon_i\leq 0)} \right)  \bm{\psi}_{i}^{ [kl^{\star}]\top} \hat{\bm{b}}_{k}^{[l^{\star}]}  \nonumber\\
&& \quad - \frac{1}{n} \sum_{i=1}^n (1{(\epsilon_i \leq 0)} - 1{(\epsilon_i \leq \bm{\psi}_i^\top (\tilde{\bm{b}}-\bm{b}^*) + \bm{z}_i^\top (\hat{\bm{\gamma}}-\bm{\gamma}^*) + u_{ni})} ) \bm{\psi}_{i}^{ [kl^{\star}]\top} \hat{\bm{b}}_{k}^{[l^{\star}]}, \nonumber
\end{eqnarray}
where $ \bm{\psi}_{i}^{ [kl^{\star}]} = (\psi_{i, k(M_n+d) + l^{\star}}, \cdots, \psi_{i, k(M_n+d) + l^{\star} + d})^\top$.
For the first term on the right hand side of the last equation of (\ref{eq:A.9}), by Conditions 1 and 3, we have: 
$$\frac{1}{n}  \sum_{i=1}^n \left(\tau-1{(\epsilon_i\leq 0)} \right)  \bm{\psi}_{i}^{ [kl^{\star}]\top} \hat{\bm{b}}_{k}^{[l^{\star}]} = O_p(n^{-1/2} \| \hat{\bm{b}}_k \|_{\bm{W}_{l^{\star}}}).$$
For the second term, since $\sup_i |\bm{\psi}_i^\top (\tilde{\bm{b}} - \bm{b}^*) + \bm{z}_i^\top (\hat{\bm{\gamma}}- \bm{\gamma}^*) + u_{ni}| = O_p(\sqrt{M_n/n})$, we have: 
\begin{eqnarray*}
&&E\left[ \left[ \frac{1}{n} \sum_{i=1}^n \left( 1{(\epsilon_i \leq 0)} - 1{(\epsilon_i \leq \bm{\psi}^\top_i (\tilde{\bm{b}}- \bm{b}^*) + \bm{z}_i^\top (\hat{\bm{\gamma}}-\bm{\gamma}^*) +u_{ni}}) \right) \bm{\psi}_{i}^{ [kl^{\star}]\top} \hat{\bm{b}}_{k}^{[l^{\star}]}   \right]^2 \right]\\
& \leq& E\left[ \left[ \frac{1}{n} \sum_{i=1}^n \left| 1{(\epsilon_i \leq C \sqrt{M_n/n})} - 1{(\epsilon_i \leq -C \sqrt{M_n/n})} \right| \times \left| \bm{\psi}_{i}^{ [kl^{\star}]\top} \hat{\bm{b}}_{k}^{[l^{\star}]}\right|   \right]^2 \right]\\
& =& E\left[  \frac{1}{n^2} \sum_{i=1}^n 1{(-C\sqrt{M_n/n} \leq \epsilon_i \leq C\sqrt{M_n/n})} \times \left| \bm{\psi}_{i}^{ [kl^{\star}]\top} \hat{\bm{b}}_{k}^{[l^{\star}]}\right|^2  \right]\\
&&\quad + \sum_{i \neq i^\prime} \frac{1}{n^2} E\left[ 1{( -C\sqrt{M_n/n} \leq \epsilon_i \leq C\sqrt{M_n/n})}   1{( -C\sqrt{M_n/n} \leq \epsilon_{i^\prime} \leq C\sqrt{M_n/n})}\right. \\
&&\quad \quad \quad  \quad \quad  \quad \quad \left.\times \left| \bm{\psi}_{i}^{[kl^{\star}]\top}  \bm{\psi}_{i^\prime}^{[kl^{\star}]}  \right| \times \|\hat{\bm{b}}_{k}^{[l^{\star}]}\|_2^2  \right]\\
& \leq &{C}{n^{-2}} \left(n M_n^{1/2}n^{-1/2} + n^2 M_n n^{-1} \right) \| \hat{\bm{b}}_k \|_{\bm{W}_{l^{\star}}}^2 =O({M_n}{n}^{-1}) \| \hat{\bm{b}}_k \|_{\bm{W}_{l^{\star}}}^2.
\end{eqnarray*}
Therefore, the second term is bounded by $O_p(\sqrt{M_n/n}\| \hat{\bm{b}}_k \|_{\bm{W}_{l^{\star}}})$, which dominates the first term. Also, since $\lambda_{\max}(\bm{V}) = O(M_n^{-1})$, we have: 
\begin{eqnarray}
\label{eq:A.10}
 \eta \sum_{k=0}^q \hat{\bm{b}}_k^\top \bm{V} \hat{\bm{b}}_k - \eta \sum_{k=0}^q \tilde{\bm{b}}_k^\top \bm{V} \tilde{\bm{b}}_k
&=& \eta \sum_{k=0}^q [(\hat{\bm{b}}_k-\tilde{\bm{b}}_k)^\top \bm{V}(\hat{\bm{b}}_k-\tilde{\bm{b}}_k)+2(\hat{\bm{b}}_k-\tilde{\bm{b}}_k)^\top \bm{V}\tilde{\bm{b}}_k ] \nonumber\\
&=& \eta \sum_{k=0}^q \hat{\bm{b}}_k^{[l^{\star}]\top} \bm{V}^{[l^{\star}]} \hat{\bm{b}}_k^{[l^{\star}]} \leq C\eta M_n^{-1} \|\hat{\bm{b}}_k^{[l^{\star}]}\|_2^2 \nonumber\\
&=&O_p(\eta\sqrt{M_n/n} \| \hat{\bm{b}}_k \|_{\bm{W}_{l^{\star}}}) ,
\end{eqnarray}
where $\bm{V}^{[l^{\star}]}$ is the submatrix of $\bm{V}$ with entries $v_{ij}, i,j = l^{\star}, \cdots, l^{\star}+d$. Since $\eta = o_p( n^{-1/2})$, together with the above discussions on (\ref{eq:A.9}) and (\ref{eq:A.10}), it follows that:
\begin{equation}
\label{eq:A.11}
\begin{aligned}
&\frac{1}{n} \sum_{i = 1}^n (  \rho_\tau(y_i - \bm{\psi}_i^\top \hat{\bm{b}} - \bm{z}_i^\top \hat{\bm{\gamma}}) -   \rho_\tau(y_i - \bm{\psi}_i^\top \tilde{\bm{b}} - \bm{z}_i^\top \hat{\bm{\gamma}})) +\eta \sum_{k=0}^q \hat{\bm{b}}_k^\top \bm{V} \hat{\bm{b}}_k - \eta \sum_{k=0}^q \tilde{\bm{b}}_k^\top \bm{V} \tilde{\bm{b}}_k \\
&= O_p (\sqrt{M_n/n}\| \hat{\bm{b}}_k \|_{\bm{W}_{l^{\star}}}).
\end{aligned}
\end{equation}

Next, we examine the functional sparse group penalty. For $I_{l^{\star}} \subset D_0^{[2]}$, we have $\|\hat{\bm{b}}_0 \|_{\bm{W}_{l^{\star}}} \geq \|\tilde{\bm{b}}_0 \|_{\bm{W}_{l^{\star}}} = 0$ and $\|\hat{\bm{b}}_k \|_{\bm{W}_{l^{\star}}} = \|\tilde{\bm{b}}_k \|_{\bm{W}_{l^{\star}}} = 0$ for all $k = 1, \cdots, q$ by the design of the sparse group penalty. And thus
\begin{equation}
\label{eq:A.12}
\sum_{k = 1}^q \sum_{l = 1}^{M_n} p_{\lambda_1} (\|\hat{\bm{b}}_k \|_{\bm{W}_{l}})= \sum_{k = 1}^q \sum_{l = 1}^{M_n} p_{\lambda_1} (\|\tilde{\bm{b}}_k \|_{\bm{W}_{l}}). 
\end{equation}
Also, for $l^\star$ such that $I_{l^\star}\subset D_0^{[2]}$, we have:
\begin{equation}
\label{eq:A.13}
\sum_{l = 1}^{M_n} p_{\lambda_2} (\|\hat{\bm{b}} \|_{\bm{W}_{l}})- \sum_{l = 1}^{M_n} p_{\lambda_2} (\|\tilde{\bm{b}} \|_{\bm{W}_{l}}) = p_{\lambda_2} (\|\hat{\bm{b}} \|_{\bm{W}_{l^{\star}}})  \geq \frac{\lambda_2}{2}\|\hat{\bm{b}} \|_{\bm{W}_{l^{\star}}}\geq \frac{\lambda_2}{2}\|\hat{\bm{b}}_k \|_{\bm{W}_{l^{\star}}},
\end{equation}
by $\|\hat{\bm{b}}\|_{\bm{W}_{l^\star}} = O(\sqrt{M_n/n})$, $M_n^{-r}/\min(\lambda_1, \lambda_2) = o(1)$, and  Condition 4. 
For $I_{l^{\star}} \subset D_k^{[2]} \text{ and } k \neq 0$, we have $\| \hat{\bm{b}}_k \|_{\bm{W}_{l^{\star}}} \geq\| \tilde{\bm{b}}_k \|_{\bm{W}_{l^{\star}}} = 0$ and $\| \hat{\bm{b}}_{k^\prime} \|_{\bm{W}_{l^{\star}}} = \| \tilde{\bm{b}}_{k^\prime} \|_{\bm{W}_{l^{\star}}}$ for $k^\prime \neq k$.
Since $\|\hat{\bm{b}}_k\|_{\bm{W}_{l^{\star}}} = O(\sqrt{M_n/n})$ and $M_n^{-r}/\min(\lambda_1, \lambda_2) = o(1)$, we have: 
\begin{equation}
\label{eq:A.14}
\sum_{k = 1}^q \sum_{l = 1}^{M_n} (p_{\lambda_1} (\|\hat{\bm{b}}_k \|_{\bm{W}_{l}}) - p_{\lambda_1} (\|\tilde{\bm{b}}_k \|_{\bm{W}_{l}})) = p_{\lambda_1} (\|\hat{\bm{b}}_k \|_{\bm{W}_{l^{\star}}}) \geq \frac{\lambda_1}{2}  \|\hat{\bm{b}}_k \|_{\bm{W}_{l^{\star}}}.
\end{equation}
For $I_{l^{\star}} \subset D_k^{[2]} \text{ and } k \neq 0$, when $\| \tilde{\bm{b}} \|_{\bm{W}_{l^{\star}}} \geq \lambda_2 \xi$, we have $\| \hat{\bm{b}} \|_{\bm{W}_{l^{\star}}} \geq \lambda_2 \xi$ and 
\begin{equation}
\label{eq:A.15}
\sum_{l = 1}^{M_n} (p_{\lambda_2}(\| \hat{\bm{b}} \|_{\bm{W}_{l}}) - p_{\lambda_2}(\| \tilde{\bm{b}} \|_{\bm{W}_{l}})) = p_{\lambda_2}(\| \hat{\bm{b}} \|_{\bm{W}_{l^{\star}}}) - p_{\lambda_2}(\| \tilde{\bm{b}} \|_{\bm{W}_{l^{\star}}}) = 0.
\end{equation}
For $I_{l^{\star}} \subset D_k^{[2]} \text{ and } k \neq 0$, when $\| \tilde{\bm{b}} \|_{\bm{W}_{l^{\star}}} < \lambda_2 \xi$, we have
\begin{equation}
\label{eq:A.16}
\sum_{l = 1}^{M_n} (p_{\lambda_2}(\| \hat{\bm{b}} \|_{\bm{W}_{l}}) - p_{\lambda_2}(\| \tilde{\bm{b}} \|_{\bm{W}_{l}})) = \lambda_2 \int_{\| \tilde{\bm{b}} \|_{\bm{W}_{l^\star}}}^{\min(\lambda_2\xi, \| \hat{\bm{b}} \|_{\bm{W}_{l^\star}})} \left(1-\frac{t}{\lambda_2\xi} \right)_+ dt \geq 0.
\end{equation}
Combining (\ref{eq:A.14}) - (\ref{eq:A.16}), we obtain that for $I_{l^{\star}} \subset D_k^{[2]}, k = 0,\cdots, q$, 
\begin{equation*}
\sum_{k = 1}^q \sum_{l = 1}^{M_n} (p_{\lambda_1}(\| \hat{\bm{b}}_k \|_{\bm{W}_{l}})-  p_{\lambda_1}(\| \tilde{\bm{b}}_k \|_{\bm{W}_{l}})) + \sum_{l = 1}^{M_n} (p_{\lambda_2}(\| \hat{\bm{b}} \|_{\bm{W}_{l}})-p_{\lambda_2}(\| \tilde{\bm{b}} \|_{\bm{W}_{l}})) \geq \frac{1}{2} \min(\lambda_1, \lambda_2) \| \hat{\bm{b}}_k \|_{\bm{W}_{l^{\star}}}.
\end{equation*}
Combining the above result with (\ref{eq:A.11}), $M_n^{-r}/\min(\lambda_1, \lambda_2) = o(1)$, and Condition 4, we prove (\ref{eq:A.8}) with probability tending to one. This completes the proof.

\clearpage
\subsection*{III. Additional numerical results}
\begin{table}[h]
\setlength{\abovecaptionskip}{0pt}
\setlength{\belowcaptionskip}{5pt}
\centering
\caption{ Scenario I under Case 2: mean (sd) based on 100 replicates. For the quantile-based methods, $\tau = 0.5$.}
\label{tab-2}
\resizebox{450pt}{60mm}
{
\begin{tabular}{cccccccc}
\hline
$n$   &           & Alt.1        & Alt.2     &Alt.3  & Alt.4   &Alt.5     & Proposed\\
\hline
    &           & \multicolumn{6}{c}{$\text{ISE}_0 (\times 10^{2})$}                                                 \\
300 & $\beta_0$ & 10.493(8.442)  & 2.012(4.538)   & 2.012(4.538)   & 8.146(5.853)   & 1.433(4.109)   & 1.433(4.109)   \\
    & $\beta_1$ & 12.101(9.790)  & 5.474(9.373)   & 3.730(6.966)   & 8.487(7.166)   & 1.683(4.468)   & 1.429(4.155)   \\
    & $\beta_2$ & 10.785(8.806)  & 4.577(8.575)   & 3.505(7.119)   & 6.868(5.097)   & 1.400(3.650)   & 1.053(2.812)   \\
500 & $\beta_0$ & 7.896(5.606)   & 1.723(4.768)   & 1.012(4.486)   & 5.330(3.898)   & 0.813(2.203)   & 0.741(2.338)   \\
    & $\beta_1$ & 6.918(5.365)   & 2.780(6.136)   & 1.545(5.678)   & 4.238(2.501)   & 0.224(0.744)   & 0.184(0.702)   \\
    & $\beta_2$ & 8.477(17.375)  & 6.062(36.792)  & 5.032(36.771)  & 4.794(3.148)   & 0.329(1.289)   & 0.538(1.900)   \\
    &           & \multicolumn{6}{c}{$\text{ISE}_1 (\times 10^{2})$}                                                 \\
300 & $\beta_0$ & 36.417(24.229) & 37.843(26.623) & 37.843(26.623) & 25.041(16.974) & 28.583(20.912) & 28.583(20.912) \\
    & $\beta_1$ & 18.563(14.815) & 18.940(15.699) & 19.167(15.704) & 11.394(10.309) & 12.484(11.503) & 14.196(13.308) \\
    & $\beta_2$ & 18.751(21.966) & 18.976(20.397) & 19.377(20.527) & 11.266(13.571) & 12.363(13.028) & 15.932(20.021) \\
500 & $\beta_0$ & 21.969(19.410) & 24.553(27.016) & 25.836(26.973) & 14.997(8.513)  & 17.765(10.305) & 18.207(11.352) \\
    & $\beta_1$ & 9.430(7.378)   & 11.483(9.900)  & 12.761(9.607)  & 7.438(6.724)   & 8.232(6.798)   & 8.674(6.646)   \\
    & $\beta_2$ & 11.958(21.889) & 14.100(30.410) & 14.360(30.392) & 7.330(6.497)   & 7.672(7.130)   & 9.201(7.945)   \\
        &           & \multicolumn{6}{c}{$\text{RMSE}_{\gamma}$}                                                           \\
300 &           & 0.135(0.105)   & 0.126(0.104)   & 0.130(0.102)   & 0.100(0.080)   & 0.091(0.071)   & 0.089(0.073)   \\
500 &           & 0.114(0.091)   & 0.112(0.087)   & 0.112(0.088)   & 0.083(0.060)   & 0.083(0.060)   & 0.087(0.062)   \\
    &           & \multicolumn{6}{c}{fTPR}                                                                           \\
300 & $\beta_0$ & 1.000(0.001)   & 0.989(0.026)   & 0.988(0.027)   & 1.000(0.000)   & 0.981(0.034)   & 0.988(0.029)   \\
    & $\beta_1$ & 1.000(0.001)   & 0.976(0.048)   & 0.990(0.036)   & 1.000(0.001)   & 0.972(0.055)   & 0.989(0.038)   \\
    & $\beta_2$ & 1.000(0.004)   & 0.981(0.053)   & 0.995(0.019)   & 1.000(0.001)   & 0.964(0.109)   & 0.990(0.043)   \\
500 & $\beta_0$ & 1.000(0.000)   & 0.990(0.029)   & 0.996(0.010)   & 1.000(0.000)   & 0.992(0.018)   & 0.991(0.020)   \\
    & $\beta_1$ & 1.000(0.000)   & 0.994(0.027)   & 0.991(0.026)   & 1.000(0.000)   & 0.993(0.028)   & 0.986(0.031)   \\
    & $\beta_2$ & 1.000(0.000)   & 0.992(0.030)   & 0.996(0.013)   & 1.000(0.000)   & 0.993(0.026)   & 0.993(0.026)  \\
    &           & \multicolumn{6}{c}{fTNR}                                                                           \\
300 & $\beta_0$ & 0.002(0.005)   & 0.490(0.326)   & 0.661(0.301)   & 0.002(0.003)   & 0.601(0.321)   & 0.728(0.314)   \\
    & $\beta_1$ & 0.003(0.004)   & 0.674(0.307)   & 0.768(0.268)   & 0.003(0.003)   & 0.830(0.233)   & 0.871(0.234)   \\
    & $\beta_2$ & 0.003(0.006)   & 0.715(0.300)   & 0.802(0.272)   & 0.003(0.002)   & 0.840(0.225)   & 0.891(0.195)   \\
500 & $\beta_0$ & 0.004(0.008)   & 0.589(0.309)   & 0.706(0.212)   & 0.003(0.004)   & 0.709(0.252)   & 0.754(0.270)   \\
    & $\beta_1$ & 0.003(0.005)   & 0.746(0.299)   & 0.861(0.179)   & 0.004(0.007)   & 0.905(0.119)   & 0.925(0.130)   \\
    & $\beta_2$ & 0.004(0.005)   & 0.753(0.303)   & 0.853(0.194)   & 0.003(0.003)   & 0.912(0.064)   & 0.912(0.159)   \\
                     \hline
\end{tabular}
}
\end{table}

\clearpage
\begin{table}[h]
\setlength{\abovecaptionskip}{0pt}
\setlength{\belowcaptionskip}{5pt}
\centering
\caption{Scenario I under Case 3: mean (sd) based on 100 replicates. For the quantile-based methods, $\tau = 0.3$.}
\label{tab-3}
\resizebox{420pt}{60mm}
{
\begin{tabular}{cccccccc}
\hline
$n$   &           & Alt.1        & Alt.2     &Alt.3  & Alt.4   &Alt.5     & Proposed\\
\hline
    &           & \multicolumn{6}{c}{$\text{ISE}_0 (\times 10^{2})$}                                              \\
300 & $\beta_0$ & 2.868(1.821)   & 0.249(0.988)   & 0.195(0.940)   & 0.472(0.311)  & 0.014(0.015)  & 0.007(0.008)  \\
    & $\beta_1$ & 5.730(3.193)   & 2.993(5.608)   & 2.583(4.649)   & 1.602(1.101)  & 0.067(0.252)  & 0.031(0.136)  \\
    & $\beta_2$ & 3.774(3.799)   & 1.328(4.369)   & 1.634(4.465)   & 0.789(0.965)  & 0.013(0.082)  & 0.066(0.476)  \\
500 & $\beta_0$ & 2.397(1.708)   & 0.031(0.177)   & 0.029(0.186)   & 0.245(0.127)  & 0.013(0.012)  & 0.008(0.024)  \\
    & $\beta_1$ & 3.791(2.452)   & 1.115(2.226)   & 1.031(1.968)   & 0.832(0.695)  & 0.008(0.020)  & 0.002(0.010)  \\
    & $\beta_2$ & 2.469(1.873)   & 0.332(1.145)   & 0.371(1.290)   & 0.345(0.259)  & 0.003(0.005)  & 0.001(0.004)  \\
    &           & \multicolumn{6}{c}{$\text{ISE}_1 (\times 10^{2})$}                                              \\
300 & $\beta_0$ & 10.983(7.331)  & 13.227(10.367) & 12.700(9.883)  & 2.060(1.673)  & 1.182(1.048)  & 1.132(0.973)  \\
    & $\beta_1$ & 19.786(23.473) & 22.760(25.114) & 22.748(27.854) & 8.944(11.601) & 8.294(11.829) & 8.141(11.990) \\
    & $\beta_2$ & 4.117(3.009)   & 4.802(5.488)   & 5.202(5.639)   & 0.979(0.668)  & 0.455(0.371)  & 0.502(0.453)  \\
500 & $\beta_0$ & 8.707(4.691)   & 9.069(6.688)   & 9.654(6.500)   & 1.045(0.593)  & 0.570(0.384)  & 0.679(0.440)  \\
    & $\beta_1$ & 12.106(14.426) & 15.345(16.454) & 15.402(16.512) & 5.110(5.191)  & 4.649(4.600)  & 4.803(4.928)  \\
    & $\beta_2$ & 3.021(1.586)   & 3.097(2.661)   & 3.726(3.473)   & 0.472(0.315)  & 0.274(0.229)  & 0.344(0.307)  \\
        &           & \multicolumn{6}{c}{$\text{RMSE}_{\gamma}$}                                                        \\
300 &           & 0.221(0.087)   & 0.221(0.083)   & 0.219(0.084)   & 0.028(0.022)  & 0.021(0.017)  & 0.022(0.017)  \\
500 &           & 0.209(0.072)   & 0.210(0.073)   & 0.210(0.073)   & 0.019(0.013)  & 0.015(0.010)  & 0.016(0.012)  \\
    &           & \multicolumn{6}{c}{fTPR}                                                                          \\
300 & $\beta_0$ & 1.000(0.000)   & 0.991(0.018)   & 0.992(0.019)   & 1.000(0.000)  & 1.000(0.000)  & 1.000(0.000)  \\
    & $\beta_1$ & 1.000(0.000)   & 0.982(0.048)   & 0.976(0.076)   & 1.000(0.000)  & 0.997(0.017)  & 0.998(0.014)  \\
    & $\beta_2$ & 1.000(0.000)   & 0.996(0.013)   & 0.991(0.022)   & 1.000(0.000)  & 1.000(0.000)  & 1.000(0.000)  \\
500 & $\beta_0$ & 1.000(0.000)   & 0.995(0.011)   & 0.993(0.017)   & 1.000(0.000)  & 1.000(0.000)  & 1.000(0.000)  \\
    & $\beta_1$ & 1.000(0.000)   & 0.979(0.057)   & 0.972(0.059)   & 1.000(0.000)  & 1.000(0.001)  & 1.000(0.001)  \\
    & $\beta_2$ & 1.000(0.000)   & 0.997(0.011)   & 0.991(0.024)   & 1.000(0.000)  & 1.000(0.000)  & 1.000(0.000)  \\
    &           & \multicolumn{6}{c}{fTNR}                                                                          \\
300 & $\beta_0$ & 0.001(0.002)   & 0.825(0.204)   & 0.829(0.182)   & 0.001(0.003)  & 0.776(0.043)  & 0.793(0.035)  \\
    & $\beta_1$ & 0.000(0.001)   & 0.718(0.262)   & 0.770(0.201)   & 0.001(0.001)  & 0.909(0.080)  & 0.937(0.049)  \\
    & $\beta_2$ & 0.001(0.001)   & 0.890(0.174)   & 0.896(0.138)   & 0.001(0.001)  & 0.933(0.030)  & 0.933(0.057)  \\
500 & $\beta_0$ & 0.001(0.001)   & 0.863(0.080)   & 0.886(0.114)   & 0.002(0.003)  & 0.785(0.035)  & 0.809(0.043)  \\
    & $\beta_1$ & 0.000(0.001)   & 0.809(0.173)   & 0.854(0.168)   & 0.001(0.001)  & 0.943(0.030)  & 0.955(0.020)  \\
    & $\beta_2$ & 0.000(0.001)   & 0.930(0.085)   & 0.938(0.094)   & 0.002(0.002)  & 0.939(0.014)  & 0.948(0.012) \\
                     \hline
\end{tabular}
}
\end{table}

\clearpage
\begin{table}[h]
\setlength{\abovecaptionskip}{0pt}
\setlength{\belowcaptionskip}{5pt}
\centering
\caption{Scenario I under Case 3: mean (sd) based on 100 replicates. For the quantile-based methods, $\tau = 0.5$.}
\label{tab-4}
\resizebox{420pt}{60mm}
{
\begin{tabular}{cccccccc}
\hline
$n$   &           & Alt.1        & Alt.2     &Alt.3  & Alt.4   &Alt.5     & Proposed\\
\hline
    &           & \multicolumn{6}{c}{$\text{ISE}_0 (\times 10^{2})$}                                           \\
300 & $\beta_0$ & 2.697(1.892)   & 0.079(0.384)   & 0.045(0.243)   & 0.485(0.301) & 0.018(0.030) & 0.007(0.006) \\
    & $\beta_1$ & 4.981(3.161)   & 2.547(5.349)   & 2.135(5.060)   & 1.628(1.227) & 0.053(0.175) & 0.065(0.245) \\
    & $\beta_2$ & 3.301(2.561)   & 0.888(2.575)   & 1.119(2.698)   & 0.794(0.705) & 0.014(0.053) & 0.006(0.044) \\
500 & $\beta_0$ & 2.038(1.510)   & 0.113(0.483)   & 0.013(0.105)   & 0.263(0.153) & 0.010(0.009) & 0.005(0.005) \\
    & $\beta_1$ & 3.587(2.548)   & 1.189(1.921)   & 0.808(1.555)   & 0.936(0.605) & 0.020(0.083) & 0.020(0.092) \\
    & $\beta_2$ & 2.181(1.581)   & 0.320(1.199)   & 0.313(1.237)   & 0.376(0.320) & 0.003(0.005) & 0.001(0.002) \\
    &           & \multicolumn{6}{c}{$\text{ISE}_1 (\times 10^{2})$}                                           \\
300 & $\beta_0$ & 9.413(6.606)   & 10.574(8.484)  & 10.629(8.146)  & 2.130(1.724) & 1.086(0.929) & 1.060(0.937) \\
    & $\beta_1$ & 15.693(18.630) & 18.776(21.860) & 17.804(21.592) & 6.430(8.397) & 5.291(7.408) & 5.119(7.025) \\
    & $\beta_2$ & 4.031(2.731)   & 4.373(3.918)   & 4.883(4.300)   & 0.902(0.638) & 0.433(0.357) & 0.445(0.375) \\
500 & $\beta_0$ & 7.600(4.344)   & 7.386(5.456)   & 8.241(5.269)   & 1.167(0.643) & 0.591(0.423) & 0.633(0.433) \\
    & $\beta_1$ & 11.086(11.832) & 12.879(12.976) & 13.214(12.947) & 4.871(5.566) & 3.898(4.967) & 4.022(4.662) \\
    & $\beta_2$ & 2.667(1.868)   & 2.783(2.685)   & 3.349(3.062)   & 0.513(0.380) & 0.237(0.178) & 0.272(0.196) \\
        &           & \multicolumn{6}{c}{$\text{RMSE}_{\gamma}$}                                                     \\
300 &           & 0.066(0.048)   & 0.065(0.047)   & 0.064(0.046)   & 0.029(0.022) & 0.037(0.030) & 0.020(0.017) \\
500 &           & 0.053(0.034)   & 0.053(0.037)   & 0.054(0.038)   & 0.020(0.015) & 0.014(0.010) & 0.015(0.011) \\
    &           & \multicolumn{6}{c}{fTPR}                                                                       \\
300 & $\beta_0$ & 1.000(0.000)   & 0.996(0.011)   & 0.995(0.012)   & 1.000(0.000) & 1.000(0.000) & 1.000(0.000) \\
    & $\beta_1$ & 1.000(0.000)   & 0.988(0.038)   & 0.985(0.033)   & 1.000(0.000) & 0.999(0.011) & 1.000(0.004) \\
    & $\beta_2$ & 1.000(0.000)   & 0.996(0.011)   & 0.992(0.019)   & 1.000(0.000) & 1.000(0.001) & 1.000(0.000) \\
500 & $\beta_0$ & 1.000(0.000)   & 0.997(0.009)   & 0.995(0.010)   & 1.000(0.000) & 1.000(0.000) & 1.000(0.000) \\
    & $\beta_1$ & 1.000(0.000)   & 0.990(0.032)   & 0.983(0.036)   & 1.000(0.000) & 1.000(0.001) & 0.999(0.004) \\
    & $\beta_2$ & 1.000(0.000)   & 0.998(0.006)   & 0.994(0.013)   & 1.000(0.000) & 1.000(0.000) & 1.000(0.000) \\
    &           & \multicolumn{6}{c}{fTNR}                                                                       \\
300 & $\beta_0$ & 0.000(0.001)   & 0.834(0.104)   & 0.860(0.132)   & 0.001(0.002) & 0.772(0.049) & 0.796(0.030) \\
    & $\beta_1$ & 0.000(0.001)   & 0.735(0.236)   & 0.802(0.187)   & 0.001(0.001) & 0.909(0.083) & 0.926(0.075) \\
    & $\beta_2$ & 0.001(0.001)   & 0.890(0.146)   & 0.901(0.126)   & 0.001(0.001) & 0.928(0.041) & 0.944(0.023) \\
500 & $\beta_0$ & 0.001(0.002)   & 0.814(0.155)   & 0.897(0.067)   & 0.002(0.002) & 0.787(0.033) & 0.813(0.029) \\
    & $\beta_1$ & 0.000(0.001)   & 0.762(0.239)   & 0.878(0.144)   & 0.001(0.001) & 0.937(0.048) & 0.947(0.036) \\
    & $\beta_2$ & 0.000(0.001)   & 0.906(0.137)   & 0.948(0.079)   & 0.002(0.002) & 0.940(0.013) & 0.950(0.009) \\
                     \hline
\end{tabular}
}
\end{table}

\clearpage
\begin{table}[h]
\setlength{\abovecaptionskip}{0pt}
\setlength{\belowcaptionskip}{5pt}
\centering
\caption{Scenario I under Case 3: mean (sd) based on 100 replicates. For the quantile-based methods, $\tau = 0.7$.}
\label{tab-5}
\resizebox{420pt}{60mm}
{
\begin{tabular}{cccccccc}
\hline
$n$   &           & Alt.1        & Alt.2     &Alt.3  & Alt.4   &Alt.5     & Proposed\\
\hline
    &           & \multicolumn{6}{c}{$\text{ISE}_0 (\times 10^{2})$}                                           \\
300 & $\beta_0$ & 2.974(2.044)   & 0.173(0.544)   & 0.071(0.374)   & 0.449(0.315) & 0.015(0.015) & 0.010(0.011) \\
    & $\beta_1$ & 5.555(4.347)   & 2.825(6.555)   & 2.698(6.350)   & 1.609(1.076) & 0.102(0.370) & 0.100(0.398) \\
    & $\beta_2$ & 3.464(2.341)   & 0.506(1.962)   & 0.994(2.512)   & 0.794(0.744) & 0.032(0.176) & 0.046(0.265) \\
500 & $\beta_0$ & 2.255(1.649)   & 0.087(0.289)   & 0.116(0.555)   & 0.244(0.126) & 0.011(0.011) & 0.005(0.005) \\
    & $\beta_1$ & 4.016(3.093)   & 1.590(3.059)   & 1.404(2.473)   & 0.845(0.652) & 0.013(0.094) & 0.007(0.064) \\
    & $\beta_2$ & 2.466(2.001)   & 0.448(1.402)   & 0.591(1.596)   & 0.362(0.262) & 0.003(0.004) & 0.001(0.001) \\
    &           & \multicolumn{6}{c}{$\text{ISE}_1 (\times 10^{2})$}                                           \\
300 & $\beta_0$ & 10.082(7.045)  & 12.193(9.736)  & 11.593(8.968)  & 1.946(1.971) & 1.238(1.487) & 1.221(1.371) \\
    & $\beta_1$ & 17.222(21.815) & 21.428(26.327) & 20.306(24.611) & 6.689(6.976) & 6.246(6.744) & 6.558(7.101) \\
    & $\beta_2$ & 4.586(2.880)   & 5.329(3.963)   & 5.427(4.056)   & 0.830(0.648) & 0.459(0.457) & 0.460(0.400) \\
500 & $\beta_0$ & 8.083(4.343)   & 8.494(5.980)   & 8.496(5.556)   & 1.015(0.566) & 0.621(0.450) & 0.667(0.455) \\
    & $\beta_1$ & 12.375(13.481) & 14.495(14.698) & 14.452(15.006) & 5.086(4.879) & 4.556(4.963) & 4.979(5.090) \\
    & $\beta_2$ & 3.081(2.565)   & 3.393(3.144)   & 3.630(3.384)   & 0.460(0.256) & 0.294(0.292) & 0.315(0.262) \\
        &           & \multicolumn{6}{c}{$\text{RMSE}_{\gamma}$}                                                     \\
300 &           & 0.190(0.087)   & 0.192(0.085)   & 0.192(0.086)   & 0.030(0.022) & 0.021(0.016) & 0.021(0.017) \\
500 &           & 0.220(0.064)   & 0.221(0.066)   & 0.221(0.066)   & 0.019(0.014) & 0.015(0.011) & 0.016(0.012) \\
    &           & \multicolumn{6}{c}{fTPR}                                                                       \\
300 & $\beta_0$ & 1.000(0.000)   & 0.990(0.021)   & 0.992(0.015)   & 1.000(0.000) & 1.000(0.000) & 1.000(0.000) \\
    & $\beta_1$ & 1.000(0.000)   & 0.974(0.107)   & 0.985(0.031)   & 1.000(0.000) & 1.000(0.002) & 1.000(0.004) \\
    & $\beta_2$ & 1.000(0.000)   & 0.990(0.022)   & 0.989(0.022)   & 1.000(0.000) & 1.000(0.000) & 1.000(0.000) \\
500 & $\beta_0$ & 1.000(0.000)   & 0.995(0.012)   & 0.996(0.009)   & 1.000(0.000) & 1.000(0.000) & 1.000(0.000) \\
    & $\beta_1$ & 1.000(0.000)   & 0.993(0.022)   & 0.990(0.022)   & 1.000(0.000) & 1.000(0.000) & 1.000(0.001) \\
    & $\beta_2$ & 1.000(0.000)   & 0.998(0.009)   & 0.994(0.016)   & 1.000(0.000) & 1.000(0.000) & 1.000(0.000) \\
    &           & \multicolumn{6}{c}{fTNR}                                                                       \\
300 & $\beta_0$ & 0.001(0.002)   & 0.817(0.201)   & 0.854(0.145)   & 0.001(0.002) & 0.770(0.041) & 0.788(0.034) \\
    & $\beta_1$ & 0.000(0.001)   & 0.737(0.292)   & 0.793(0.190)   & 0.001(0.001) & 0.908(0.083) & 0.925(0.080) \\
    & $\beta_2$ & 0.001(0.002)   & 0.897(0.167)   & 0.884(0.157)   & 0.001(0.002) & 0.924(0.043) & 0.933(0.055) \\
500 & $\beta_0$ & 0.001(0.001)   & 0.820(0.151)   & 0.837(0.172)   & 0.002(0.004) & 0.786(0.037) & 0.812(0.030) \\
    & $\beta_1$ & 0.001(0.001)   & 0.766(0.263)   & 0.828(0.201)   & 0.001(0.001) & 0.943(0.024) & 0.953(0.028) \\
    & $\beta_2$ & 0.001(0.002)   & 0.892(0.149)   & 0.917(0.116)   & 0.002(0.002) & 0.942(0.013) & 0.951(0.009) \\
                     \hline
\end{tabular}
}
\end{table}

\clearpage
\subsection*{IV. Additional data analysis results}

We conduct exploratory regression analysis with the proposed penalty. For the lack-of-fit, we consider the mean-based (Alt.3) and the proposed quantile-based. With the proposed approach, we consider $\tau=0.3, 0.5, 0.7$. In the left panel of Figure \ref{Fig:5}, we plot the estimated densities. It is observed that the residuals are left-skewed, which suggests the sensibility of quantile-based analysis. Different quantiles lead to different results, which has been commonly observed in the literature. In addition, the {mean} estimation is closer to the proposed estimation with $\tau=0.3$, compared to the other two quantile values. 

\begin{figure}[h]
\centering
\includegraphics[ height=2.2 in, width=3.5 in]{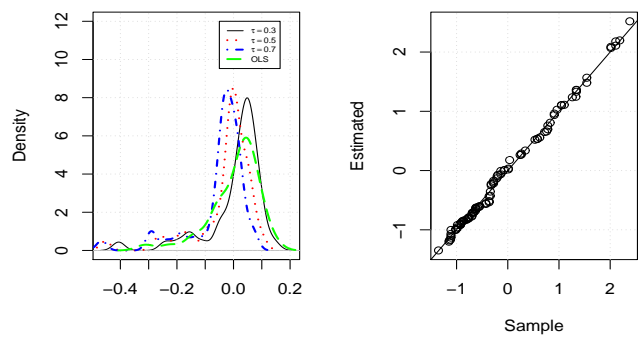}
\caption{Left: estimated densities of residuals from Alt.3 and the proposed method with $\tau=0.3, 0.5, 0.7$. Right: Lack-of-fit diagnostic QQ plot.}
\label{Fig:5}
\end{figure}

We also conduct model diagnostics using a QQ plot to intuitively assess model fitting. Specifically, we first randomly generate $\breve\tau$ from the uniform distribution on [0,1]. 
We then fit data using the proposed method with quantile $\breve\tau$ and obtain estimator $(\hat{\bm{b}}(\breve{\tau}), \hat{\bm{\gamma}}(\breve{\tau}))$. Next, we generate the response from the model $\breve{y} =\bm{\psi}^\top \hat{\bm{b}}(\breve{\tau}) + \bm{z}^\top \hat{\bm{\gamma}}(\breve{\tau}) $, where $(\bm{\psi}, \bm{z})$ are randomly selected from the original data. We repeat this process 100 times and obtain a sample of 100 simulated fat values. The right panel of Figure \ref{Fig:5} gives the QQ plot for the simulated and observed fat contents. Most points are very close to the 45-degree line, which suggests satisfactory model fitting.

\end{document}